\documentclass[conference]{IEEEtran}
\IEEEoverridecommandlockouts

\usepackage[utf8]{inputenc} 
\usepackage[T1]{fontenc}    
\usepackage{hyperref}       
\usepackage{url}            
\usepackage{booktabs}       
\usepackage{amsfonts}       
\usepackage{nicefrac}       
\usepackage{microtype}      
\usepackage{xcolor}         
\usepackage{cite}

\usepackage{algorithmic}
\usepackage{pifont}
\usepackage[linesnumbered, ruled]{algorithm2e}

\usepackage{microtype}
\usepackage{graphicx}
\usepackage{subfigure}
\usepackage{booktabs} 
\usepackage{bm}
\usepackage{enumitem}

\usepackage{multirow}
\usepackage{tabularx}
\usepackage{booktabs}
\usepackage{diagbox} 
\usepackage{multirow} 
\usepackage{makecell} 
\usepackage{longtable} 
\usepackage{makecell} 
\usepackage{colortbl}

\usepackage{wrapfig}

\usepackage{amsmath}
\usepackage{amssymb}
\usepackage{mathtools}
\usepackage{amsthm}
\usepackage{enumitem}

\allowdisplaybreaks[4]

\hypersetup{hidelinks,colorlinks=true,
linkcolor=black,anchorcolor=black,citecolor=darkgray}

\newtheorem{lemma}{Lemma}
\newtheorem{theorem}{\bf Theorem}

\newtheorem{remark}{\bf Remark}
\newtheorem{assumption}{\bf Assumption}

\def\BibTeX{{\rm B\kern-.05em{\sc i\kern-.025em b}\kern-.08em
    T\kern-.1667em\lower.7ex\hbox{E}\kern-.125emX}}

\begin{document}

\title{Semantic Communication in Dynamic Channel Scenarios: Collaborative Optimization of Dual-Pipeline Joint Source-Channel Coding and Personalized Federated Learning}

\author{
\IEEEauthorblockN{\small1\textsuperscript{st} Xingrun Yan}
\IEEEauthorblockA{\small\textit{School of Cyberspace Science and Technology} \\
\textit{Beijing Institute of Technology}\\
Beijing, China \\
xingrunyan@bit.edu.cn}
\and
\IEEEauthorblockN{\small2\textsuperscript{nd} Shiyuan Zuo}
\IEEEauthorblockA{\small\textit{School of Cyberspace Science and Technology} \\
\textit{Beijing Institute of Technology}\\
Beijing, China \\
zuoshiyuan@bit.edu.cn}
\and
\IEEEauthorblockN{\small3\textsuperscript{rd} Yifeng Lyu}
\IEEEauthorblockA{\small\textit{Institute of Computer Application Technology} \\
\textit{Norinco Groups}\\
Beijing, China \\
yifenglyu@hotmail.com}
\and
\IEEEauthorblockN{\small4\textsuperscript{th} Rongfei Fan}
\IEEEauthorblockA{\small\textit{School of Cyberspace Science and Technology} \\
\textit{Beijing Institute of Technology}\\
Beijing, China \\
fanrongfei@bit.edu.cn}
\and
\IEEEauthorblockN{\small5\textsuperscript{th} Han Hu}
\IEEEauthorblockA{\small\textit{School of Information and Electronics} \\
\textit{Beijing Institute of Technology}\\
Beijing, China \\
huhan627@gmail.com}
}

\maketitle

\begin{abstract}
Semantic communication is designed to tackle issues like bandwidth constraints and high latency in communication systems. 
However, in complex network topologies with multiple users, the enormous combinations of client data and channel state information (CSI) pose significant challenges for existing semantic communication architectures. 
To improve the generalization ability of semantic communication models in complex scenarios while meeting the personalized needs of each user in their local environments, we propose a novel personalized federated learning framework with dual-pipeline joint source-channel coding based on channel awareness model (PFL-DPJSCCA).
Within this framework, we present a method that achieves zero optimization gap for non-convex loss functions. 
Experiments conducted under varying SNR distributions validate the outstanding performance of our framework across diverse datasets.
\end{abstract}

\section{Introduction}
\label{sec:intro}

With the continuous advancement of wireless communication network technologies and the widespread adoption of various data-intensive applications such as AR/VR multimedia, traditional communication systems are facing significant challenges in supporting massive data transmission.
Concurrently, as the developing of the sixth-generation (6G) network, the integration of satellite internet into terrestrial communication systems becomes increasingly feasible. 
However, the satellite-to-ground transmission links are inherently constrained by limitations in bandwidth and latency. 
To address these existing and potential challenges, deep learning-based joint source-channel coding (Deep JSCC) has surfaced as a promising approach, serving as a method to realize Semantic Communication (SC).
Leveraging the rapid evolution of deep learning, this approach replaces independent source and channel coding modules with neural network-based interactive and jointly trained systems, demonstrating substantial potential for development.
Particularly in image transmission tasks, Deep JSCC exhibits comparable performance and superior efficiency when compared with advanced image codec (e.g., JPEG/JPEG2000/BPG) combined with channel coding schemes like Low-Density Parity-Check (LDPC) codes.

In recent years, higher performance demands have driven advancements in JSCC. 
\cite{9414037} proposed integrating a Channel State Information (CSI) estimation module into the JSCC decoder to adaptively process images based on CSI. 
However, this approach requires extensive training and large parameter spaces for unpredictable channels. 
\cite{10094735} enhanced Deep JSCC by embedding Swin Transformer and a shared CSI mechanism in both encoder and decoder, improving channel generalization but causing disconnection between semantic and channel codec to make it challenging to achieve a globally optimal solution. 
\cite{9438648} introduced a attention feature model blended between feature extraction modules, enhancing adaptability to random channels but increasing complexity and latency. 
In contrast, \cite{9878262} incorporated traditional modules like demodulation and quantization into semantic communication, enabling adaptive CSI optimization.

In modern communication scenarios, network topologies typically feature multi-user access, where multiple client nodes connect to a central node, resembling noval network topologies such as edge computing and self-organizing networks. 
However, in practical training and deployment, CSI exhibits continuous and dynamic variations, posing challenges for adaptive joint codec. 
To address the issue of channel generalization, simply discretizing the continuous variable of CSI would inevitably lead to performance degradation.
Moreover, in multi-user environments, the combination of multi-user channel information and complex semantic encoding results in an exponential expansion of the channel state space.

Existing research has proposed various optimization methods for JSCC communication systems, but there is still room for improvement. 
Firstly, \cite{10584091} optimized channel utilization through channel resource reallocation and model-based adaptive channel gain adjustment modules. However, in dynamic CSI environments, the channel perception module needs to work jointly with semantic feature extraction to enhance the model's generalization capability. 
Secondly, \cite{10419624} improved local model performance in Federated Learning (FL) through distillation techniques and mixed training, while \cite{10531097} introduced a joint contrastive learning framework to better capture local data characteristics, improving global model performance and mitigating semantic imbalance. 
However, in multi-user network topologies, these methods do not fully address the challenges of dynamic CSI, making it difficult to balance global generalization capabilities with local personalized requirements when faced with vast combinations of dynamic CSI and transmission information. 
Finally, \cite{10559618} explored FL aggregation methods and proposed loss-weighted aggregation to accelerate convergence, but it lacks theoretical proofs to ensure the algorithm's theoretical validity.

To address these limitations, we propose a Personalized Federated Learning Framework with Dual-Pipeline Joint Source-Channel Coding based on Channel Awareness (PFL-DPJSCCA).
In this framework, only the encoder is equipped with an auxiliary pipeline.
On the auxiliary pipeline, we simultaneously acquire CSI and intermediate processing results from the encoder's main pipeline to generate noise-resistant masks. 
Each stage of the auxiliary pipeline dynamically adjusts the mask by inheriting outputs from the previous stage and perceiving the mapping relationships between adjacent stages of the main pipeline. 
In the encoder, the integration of multi-scale and multi-level mask generation facilitates dense feature adjustment within the main pipeline. 
The collaborative optimization significantly enhances the JSCC model's generalization capability across diverse CSI conditions.
In the decoder, the received information is further denoised through a preprocessing operation with prior knowledge and then provided to the task-oriented functional modules.
Furthermore, we integrate personalized federated learning (PFL), where local CSI perception and processing parameters are updated locally, while task-oriented main pipeline parameters are uploaded for global model updates.

The main contributions of this paper are summarized as follows:
\begin{enumerate}
    \item {\bf Algorithmically}, we propose a PFL integrated with a channel-aware JSCC model, termed PFL-DPJSCCA. 
    By enhancing information interaction between pipelines, PFL-DPJSCCA improves the model's generalization capability for diverse channel conditions. 
    Within the PFL, parameters for CSI perception and processing are iteratively updated locally, while the main link parameters participate in global updates, enabling adaptive optimization for local environmental conditions.

    \item {\bf Theoretically}, we establish convergence guarantees for the proposed framework under non-convex and smooth loss function assumptions. 
    Specifically, with a learning rate $\mathcal{O}(T^{q})$ where $q \in (-1,0)$, PFL-DPJSCCA achieves a convergence rate of $\mathcal{O}(T^{1-q})$. 
    More importantly, we prove that the framework attains zero optimization gap under described conditions. 

    \item {\bf Empirically}, we conduct comprehensive experiments across diverse CSI conditions and communication scenarios. 
    In complex environments, particularly under more adverse conditions, the experimental results demonstrate that PFL-DPJSCCA outperforms baseline methods by improving test accuracy by 0.04\%-38.10\%, while maintaining superior robustness to channel variations.
    The experimental findings also validate the effect of our proposed dual-pipeline architecture and adaptive interaction mask mechanisms.
\end{enumerate}

\section{Preliminary Definition}
The problem setup is systematically formulated by integrating the principles of  FL and SC, establishing a comprehensive foundation for subsequent analysis and implementation.

\subsection{PFL Framework:} 

In a Personalized Federated Learning (PFL) system, $ N $ clients form the set $\mathcal{N} \triangleq \{1, 2, ..., N\}$. Each client $ n $ has a local dataset $\mathcal{D}_n$ with $ D_n $ elements, where $\bm{z}_{n,i} = \{\bm{x_{n,i}}, y_{n,i}\}$ is the $ i $-th data point, $\bm{x_{n,i}} \in \mathcal{R}^{\text{in}}$ is the input, and $ y_{n,i} \in \mathcal{R}^{\text{out}} $ is the output.

PFL trains global parameters $\bm{u}$ and local personalized parameters $\bm{v}_n$, where $\bm{u}$ is updated globally and $\bm{v}_n$ is optimized locally. Our goal is to train these parameters by minimizing the following loss function:
\begin{small}
\begin{align}
F(\bm{u}, \bm{v}) \triangleq \sum_{n=1}^N \gamma_n f_n(\bm{u}, \bm{v}_n),
\end{align}
\end{small}
where $\gamma_n = \frac{D_n}{D_A}$, $ D_A = \sum_{n=1}^N D_n $, and we define 
\begin{small}
\begin{align}
f_n(\bm{u}, \bm{v}_n) = \frac{1}{D_n} \sum_{\bm{z}_{n,i} \in \mathcal{D}_n} f(\bm{u}, \bm{v}_n, \bm{z}_{n,i})  
\end{align}
\end{small}
which is the local loss for client $ n $.

For a random subset $\xi_n \subseteq \mathcal{D}_n$, the stochastic loss is:
\begin{small}
\begin{align}
f(\bm{u}, \bm{v}_n, \xi_n) = \frac{1}{|\xi_n|} \sum_{\bm{z}_{n,i} \in \xi_n} f(\bm{u}, \bm{v}_n, \bm{z}_{n,i}).
\end{align}
\end{small}

So, the training task can be reformulated as the following optimization problem:
\begin{small}
\begin{align}
\min_{\bm{u}, \bm{v}} F(\bm{u}, \bm{v}) = \sum_{n=1}^N \gamma_n f_n(\bm{u}, \bm{v}_n).
\end{align}
\end{small}

In PFL, the parameter server (PS) iteratively interacts with clients to exchange gradients or model parameters. Clients update $\bm{u}$ globally and $\bm{v}_n$ locally. The gradients are $\nabla_{\bm{u}} f_n(\bm{u}, \bm{v}_n)$, $\nabla_{\bm{v}_n} f_n(\bm{u}, \bm{v}_n)$, $\nabla_{\bm{u}} F(\bm{u}, \bm{v})$, and $\nabla_{\bm{v}} F(\bm{u}, \bm{v})$.
PFL balances global and personalized models, adapting better to each client's local data distribution.

\subsection{Semantic Communication Framework:} 
Following the common SC architecture design, our proposed algorithm involves two main functional modules: encoder and decoder. 
The encoder reduces input data into a compact, dense feature matrix, essential for efficient transmission. 
And the decoder then reconstructs these features into meaningful outputs for task execution.
In this study, we focus on image transmission and reconstruction as the standard tasks, which are widely recognized as fundamental benchmarks within the field of semantic communication.
Given that the task is image reconstruction, the input is defined as an image $\mathcal{X}$. 
This image undergoes transformation by the encoder, yielding an output referred to as $Emb_S$, which is transmitted from the sender.
So, we have
\begin{align}
    Emb_S =\textup{Encoder}(\mathcal{X}) 
\end{align}
where we denote $\textup{Encoder}$ represent the encoder module.
As $Emb_S$ propagates through the communication channel, it encounters inevitable transmission challenges including noise interference and channel fading. 
Consequently, the information received at the receiver is defined as $Y$, which encapsulates the effects of these transmission adversities.
Then, it can be explained that 
\begin{align}
    Y=H \cdot Emb_{S}+N
\end{align}
where $H$ denotes the fading coefficient at the receiver and $N$ is the environment noise.
The received signal, once demodulated, enables the recovery of the transmitted information with high fidelity.
We can receive
\begin{align}
    Emb_{R}=\left(H^TH\right)^{-1}H^{T}Y=Emb_S+\hat{N}
\end{align}
where $Emb_{R}$ is the estimated symbol after the transformed operation, and $\hat{N}$ is the modified environment noise.
For image reconstruction tasks, by processing the input through the decoder, we can obtain the reconstructed image $\hat{\mathcal{X}}$.
\begin{align}
    \hat{\mathcal{X}} = \textup{Decoder}(Emb_R)
\end{align}
where $\textup{Decoder}$ refers to the decoder module.
The most commonly used loss function is the Mean Squared Error (MSE) for the learning task, which can be expressed as :
\begin{align}
\label{e: mse loss}
    L_{mse}(\mathcal{X},\hat{\mathcal{X}})=\textup{MSE}(\mathcal{X},\hat{\mathcal{X}})
\end{align}
\section{Framework of PFL Based on DPJSCCA}
\label{sec: framework of cajscc}
In this section, we introduce the framework of PFL with DPJSCCA model. 
The DPJSCCA primarily consists of two key components: semantic encoder and semantic decoder. 
In our design, we have intentionally de-emphasized the presence of channel encoder or decoder to achieve two main objectives: (1) to accommodate diverse requirements of communication coding schemes, and (2) to establish a more robust correlation between channel awareness function and image feature extraction work.

\begin{figure}
    \centering
    \includegraphics[width=\linewidth]{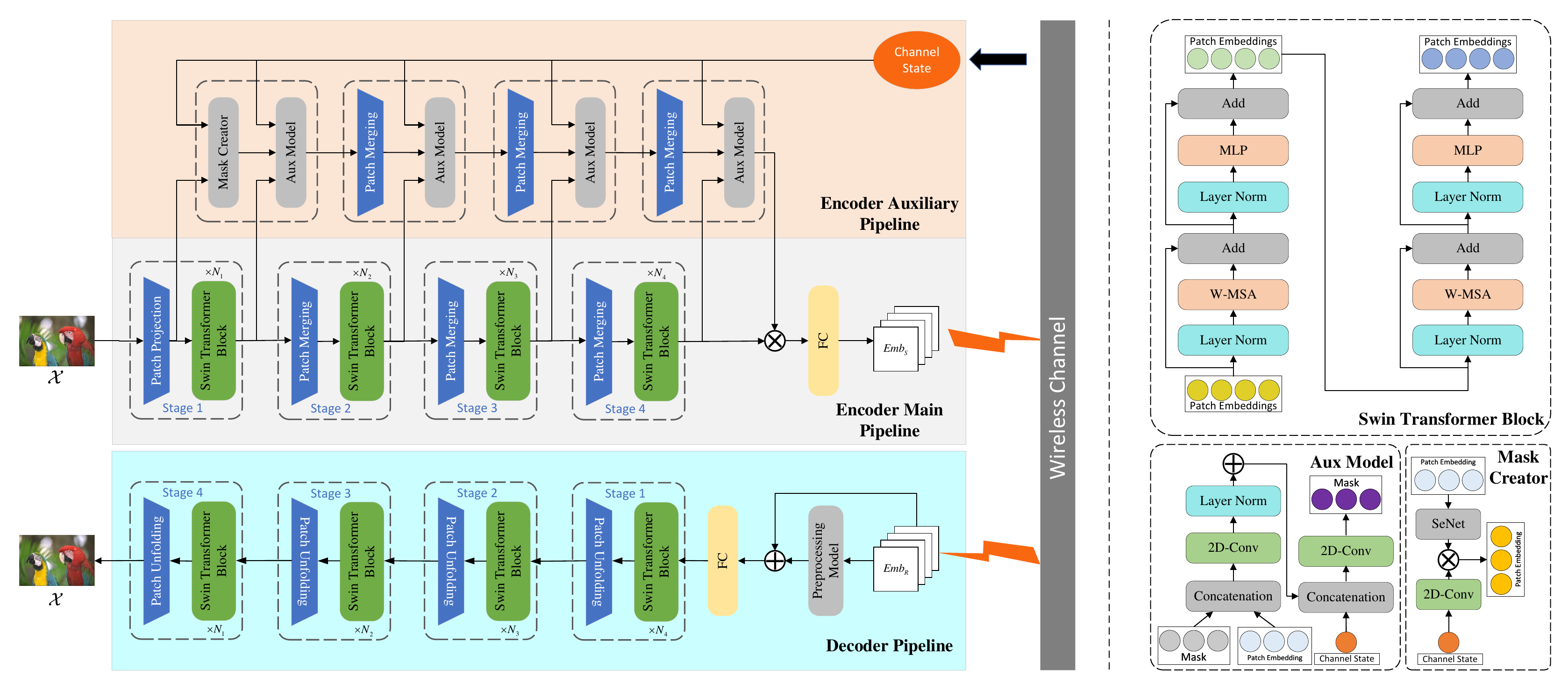}
    \caption{Proposed DPJSCCA model architecture (left) and detailed components: Swin Transformer Blocks, Aux Model, and Mask Creator (right).}
    \label{fig: model pic}
\end{figure}

\subsection{Semantic Encoder:}
We employ Swin Transformer Blocks as the core module of the semantic encoder-decoder in the main pipeline, which is a widely recognized standard design approach \cite{10094735,10531097}.
To further enhance system performance, we introduce an auxiliary pipeline specifically designed to perceive and adapt to dynamic CSI.
As shown in Fig. \ref{fig: model pic}, the input image $\mathcal{X}$ is first processed by a convolution network for dimension reduction, generating a feature map containing dense feature, which serves as the input for the main pipelines. 
In the main pipeline, the input features undergo deep processing through multiple Swin Transformer Blocks, ultimately producing a dense feature representation rich in semantic information.
\begin{align}
    \mathcal{X}' = \textup{SW-MSA}(\mathcal{X})
\end{align}
where SW-MSA denote that the Swin Transformer Block.
\begin{align}
    Mask = \textup{Auxiliary}(\mathcal{X},Snr)
\end{align}
where SW-MSA denote that the Swin Transformer Block.

Considering that the image information undergoes complex transformations through multiple Swin Transformer Blocks in the main pipeline, we establish an information gain channel between the main and auxiliary pipelines to ensure that the generated mask can effectively enhance the image's noise resistance.
The information gain channel allows the auxiliary pipeline to perceive the processing of each Swin Transformer Block.
Notably, the auxiliary pipeline also requires CSI as an additional input, where we use the Signal-to-Noise Ratio (SNR) as a quantitative metric for CSI evaluation.
Finally, the system performs inner product between the mask output from the auxiliary pipeline and the dense features generated by the main pipeline, completing the feature enhancement.

The innovation of this dual-pipeline architecture lies in its channel-aware mechanism through the auxiliary pipeline, enabling the system to adaptively adjust feature representations, thereby improving the model's robustness in complex channel environments. 
Simultaneously, the information interaction between the main and auxiliary pipelines ensures deep coordination between the channel awareness process and feature extraction, achieving superior performance.

\subsection{Semantic Decoder:}
In the semantic decoder, the model is divided into a noise secondary processing module and multiple Swin Transformer Blocks.
Under the perceiving the input from the auxiliary pipeline, we preprocess the input $Emb_R$ through secondary processing to facilitate deeper and more effective denoising, thereby obtaining further refined dense features.
\begin{align}
    Emb_{R}' = \textup{Preprocess}(Emb_R)
\end{align}
where Preprocess refers to the model aimed at preprocessing communication information.
These features $Emb_{R}'$ are then subjected to image reconstruction tasks through reverse multiple Swin Transformer Blocks.
\begin{align}
    \hat{\mathcal{X}} = \textup{SW-MSA}(Emb_R')
\end{align}
The distance between $\mathcal{X}$ and $\hat{\mathcal{X}}$ is measured using MSE, and the model is iteratively updated based on the loss function defined by equation (\ref{e: mse loss}).

\subsection{Personalized Federated Learning}
We aim to combine SC module and the PFL framework to adapt to the communication system where multiple clients correspond to an independent server.
To achieve our objectives, we integrate the SC model into local clients, where it undergoes multiple rounds of local iterative training. 
The locally updated models periodically interact with the PS. 
Upon receiving the model parameters from the clients, the PS updates the global model through an aggregation operation. 
Subsequently, the PS broadcasts the updated global model back to all clients, preparing for the next round of global iteration.
    To enhance the DPJSCCA model's generalization for dynamic CSI, we employ personalized parameter techniques by retaining the semantic encoder's auxiliary pipeline parameters and the preprocessing of semantic decoder, improving adaptability to varying channel conditions.
Below, we provide a full description of PFL-DPJSCCA (see Algorithm \ref{alg: pfl algorithm}), with its crucial steps explained in details as follows.

\begin{algorithm}[tb]
\small
\caption{PFL-DPJSCCA Algorithm} 
\label{alg: pfl algorithm}
\begin{algorithmic}[1]
    \STATE {\bfseries Input:} Initial global model parameter $u^{0}$, local model parameters $\{v_n^{0}\}_{n=1}^N$, number of communication rounds $T$, number of clients $N$, number of local steps $\tau$, step sizes $\eta_v, \eta_u$.
    \STATE {\bfseries Output:} Final global model parameter $u^{T}$ and local model parameters $\{v_n^{T}\}_{n=1}^N$.
    \STATE {\% \% \bf Initialization}
    \STATE Server initializes global model $u^{0}$ and each device $n$ initializes local model $v_n^{0}$.
    \FOR{$t=0$: $T-1$}
        \STATE Server broadcasts the current global model $u^{t}$ to each device in $\mathcal{N}$.
        \STATE {\%\% \bf Local Model Update on Devices}
        \FOR{each device $n \in \mathcal{N}$ in parallel}
            \STATE Initialize local model $\bm{v}_{n, 0}^t = \bm{v}_n^{t}$ and global model $\bm{u}_{n, 0}^t = \bm{u}^{t}$.
            \FOR{$k=0$: $\tau-1$}
                \STATE Update local model: $\bm{v}_{n, k+1}^t = \bm{v}_{n, k}^t - \eta_v {G}_{n,v}\left(\bm{u}_{n, k}^t, \bm{v}_{n, k}^t, \xi_{n,k}^t\right)$.
                \STATE Update global model copy: $\bm{u}_{n, k+1}^t = \bm{u}_{n, k}^t - \eta_u {G}_{n,u}\left(\bm{u}_{n, k}^t, \bm{v}_{n, k}^t, \xi_{n,k}^t\right)$.
            \ENDFOR
            \STATE Set $\bm{v}_n^{t+1} = \bm{v}_{n, \tau}^t$ and $\bm{u}_n^{t+1} = \bm{u}_{n, \tau}^t$.
            \STATE Device $n$ sends updated global model $\bm{u}_n^{t+1}$ back to the server.
        \ENDFOR
        \STATE {\%\% \bf Global Model Aggregation}
        \STATE Server aggregates the received local updates: $\bm{u}^{t+1}=\bm{u}^{t}-\frac{\eta_u}{N}\sum_{n=1}^{N} \sum_{k=0}^{\tau}{G}_{n,u}\left(\bm{u}_{n, k}^t, \bm{v}_{n, k}^t, \xi_{n,k}^t\right)$.
    \ENDFOR
    \STATE {\%\% \bf Output Final Model}
    \STATE Output the final global model $u^{T}$ and local model parameters $\{v_n^{T}\}_{n=1}^N$.
\end{algorithmic}
\end{algorithm}

{\bf Local Updating:}
In the $t$th round of iteration, after receiving the global model parameter $\bm{u}^t$ broadcast by the PS, all clients $\mathcal{N}$ initial their own personalized parameters $\bm{v}_{n}^{t}$.
Then each client $n$, where $n \in \mathcal{N}$, randomly samples subdataset $\xi_n^t$ from their dataset $\mathcal{D}_n$ to calculate their local training gradient $\nabla F_n(\bm{u}^{t}, \bm{v}^{t}_{n}, \xi_n^t)$.
All parameters will jointly function in the local training process and be simultaneously updated through gradient backward.
Let ${G}_n$ represent the vector of the local training gradient uploaded to the PS by client $n$.
Here, ${G}_{n,u}$ and ${G}_{n,v}$ represent the partial derivatives of the local loss function with respect to the global parameters $u$ and the local personalized parameters $v$, respectively.
For personalized parameters, we introduce contrastive learning to enhance the efficiency of core feature extraction from input and improve the stability of dual-channel training.
We use common Information Noise Contrastive Estimation (InfoNCE) to achieve our goals, that is,
\begin{align}
    L_{cl} = -\log\frac{\exp(f(x)^T f(x^+) / \tau)}{\sum_{i=1}^N \exp(f(x)^T f(x_i) / \tau)}
\end{align}
So, the total loss is combined with MSE and InfoNCE, we have
\begin{align}
    L_{total}=L_{mse}+L_{cl}
\end{align}

{\bf Aggregation and Broadcasting:}
In the $t$th round of iteration, after receiving the vectors ${G}_{n,u}\left(\bm{u}^{t}, \bm{v}^{t}_{n}, \xi_n^t\right)$ from all clients $n \in \mathcal{N}$, the PS updates the global model parameter ${u}^{t+1}$ using the learning rate $\eta_{u}$, as given by
\begin{align}
    \bm{u}^{t+1}=\bm{u}^{t}-\frac{\eta_u}{N}\sum_{n=1}^{N} \sum_{k=0}^{\tau}{G}_{n,u}\left(\bm{u}_{n, k}^t, \bm{v}_{n, k}^t, \xi_{n,k}^t\right)
\end{align}
which we assume that the weights of the parameters from each client are equal during aggregation, is introduced to facilitate our discussion of its properties.
Then the PS broadcasts the global model parameter ${u}^{t+1}$ to all clients in preparation for the calculation in the $t+1$th iteration. 

\section{Theoretical Results}
\label{sec: convergence analysis}
In this section, we theoretically analyze the convergence performance of PFL with DPJSCCA. 
Below, we first present the necessary assumptions.
Then analytical results of convergence on our proposed method are presented based on assumptions. 
All the proofs are deferred to Appendix B, which are available in the supplementary materials.

\subsection{Necessary Assumptions}
First we state some general assumptions in the work.
\begin{assumption}
\label{ass: main smooth}
($L$-smooth). The loss function $F_n(\bm{u}_n, \bm{v}_n)$ is continuously differentiable for each client $n=1, \dots, N$.
There exists a constant $L > 0$ such that for $\bm{u}_n, \bm{v}_n$, the following inequality holds:
\begin{small}
\begin{align}
    \nonumber
    F_n\left(\bm{u}_n^{t}, \bm{v}_n\right) 
    &\le F_n\left(\bm{u}_n^{t+1}, \bm{v}_n\right) + \frac{L}{2}\|\bm{u}_n^{t+1} - \bm{u}_n^{t}\|^2 \\
    &+ \left\langle \nabla F_n\left(\bm{u}_n^{t+1}, \bm{v}_n\right), \bm{u}_n^{t+1} - \bm{u}_n^{t} \right\rangle.
\end{align}
\end{small}
Consequently, the following inequality holds:
\begin{small}
\begin{align}
    \|F_n\left(\bm{u}_n^{t+1}, \bm{v}_n\right) - F_n\left(\bm{u}_n^{t}, \bm{v}_n\right)\| &\le L\|\bm{u}_n^{t+1} - \bm{u}_n^{t}\|.
\end{align}
\end{small}
To characterize the smoothness of $F_n$ with respect to both $\bm{u}$ and $\bm{v}$, we define positive constants $L_u, L_v,L_{uv},L_{vu}$ to satisfy the above inequalities.

\end{assumption}

\begin{assumption}
\label{ass: main Bounded Inner Error}
(Bounded Inner Error). For each client \( n = 1, \dots, N \), a subset of dataset \( \mathcal{D}_i \) is randomly selected, denoted as \( \xi_n \). Define
\begin{equation}
    F_{n}\left(\bm{u}_n, \bm{v}_n, \xi_n\right) \triangleq \frac{1}{|\xi_n|} \sum_{\bm{z}_n \in \xi_n} f\left(\bm{u}_n, \bm{v}_n, \bm{z}_n\right).
\end{equation}
The stochastic gradients \( G \) are unbiased, i.e.,
\begin{small}
\begin{align}
    \nabla_{u} F_{n}\left(\bm{u}_n, \bm{v}_n\right) &= \mathbf{E}\left[ G_{n,u}\left(\bm{u}_n, \bm{v}_n, \xi_n\right)\right], \\
    \nabla_{v} F_{n}\left(\bm{u}_n, \bm{v}_n\right) &= \mathbf{E}\left[ G_{n,v}\left(\bm{u}_n, \bm{v}_n, \xi_n\right)\right].
\end{align}
\end{small}
Furthermore, the inner error of stochastic gradients is bounded. There exist constants \( \sigma_{u} \) and \( \sigma_{v} \) such that
\begin{small}
\begin{align}
    \|G_{n,u}\left(\bm{u}_n, \bm{v}_n\right) - \nabla_{u}F_n\left(\bm{u}_n, \bm{v}_n\right)\|^2 &\le \sigma^2_{u}, \\
    \|G_{n,v}\left(\bm{u}_n, \bm{v}_n\right) - \nabla_{v}F_n\left(\bm{u}_n, \bm{v}_n\right)\|^2 &\le \sigma^2_{v}.
\end{align}
\end{small}
\end{assumption}

\begin{assumption}
\label{ass: main Bounded Outer Error}
(Bounded Outer Error). For \( F_n\left(\bm{u}_n, \bm{v}_n\right) \), the variance of stochastic gradients is bounded by constant \( \delta \), i.e.,
\begin{small}
\begin{align}
&\frac{1}{N} \sum_{n=1}^{N}\|\nabla_{u} F_n\left(\bm{u}_n,\bm{v}_n\right) - \nabla_{u} F\left(\bm{u}_n,V\right)\|^2 \le \delta^2
\end{align}
\end{small}
This inequality serves to characterize the degree of detriment to global convergence inflicted by each client as a consequence of the data heterogeneity issue. 
\end{assumption}

\subsection{Convergence Analysis Results}
\begin{theorem}
\label{thm: nonconvex convergence}
With the holding of Assumptions \ref{ass: main smooth} to \ref{ass: Bounded Outer Error}, for the set of constant $T, N, \tau, L_u, L_v, L_{uv},L_{vu},\delta,\sigma_u^2, \sigma_v^2$ is same with the above assumptions and definitions. 
And we define the learning rate for $\bm{u}$ and $\bm{v}$ is that $\eta_{u} = \frac{c}{\tau L_{u}\left(1+\rho^2\right)}$ and $\eta_{v} = \frac{c}{\tau L_{v}}$, where $c \in \left(0,\min{\left\{\frac{1}{\sqrt{6}\max\left\{1,\frac{L_{vu}L_{uv}}{L_{u}L_{v}}\right\}},\frac{\min{\left\{L_{u},L_{v}\right\}}}{6\max\left\{L_{u},L_{v}\right\}+L_{uv}}\right\}}\right]$   there is 

\begin{align}
\nonumber
&\frac{1}{T}\sum_{t = 0}^{T}\left\|\nabla_{u}F(\bm{u}^{t},V^{t})\right\|^2+\frac{1}{NT}\sum_{t = 0}^{T}\sum_{n = 1}^{N}\left\|\nabla_{v}F(\bm{u}_{n}^{t},\bm{v}_{n}^{t})\right\|^2\\
&\leq\frac{F^* - F(\bm{u}^0,V^0)}{T\min\{\lambda_1,\lambda_2\}}+\frac{3(L_u + L_{uv})c^2}{2L_u^2\min\{\lambda_1,\lambda_2\}}\delta^2\\
\nonumber
&\qquad\qquad+\frac{\beta_u}{\min\{\lambda_1,\lambda_2\}}\sigma_{u}^2+\frac{\beta_v}{\min\{\lambda_1,\lambda_2\}}\sigma_{v}^2
\end{align}
which we denote that
\begin{align}
    \beta_u = \frac{(L_u + L_{uv})c^2}{2L_u^2}+\frac{5(e^2 - 1)c^3}{4\tau^2L_u}+\frac{5(e^2 - 1)L_{vu}^2c^2}{4\tau^2L_u^2L_v} \\
    \beta_v = \frac{(L_v + L_{uv})c^2}{2L_v^2}+\frac{5(e^2 - 1)c^3}{4\tau^2L_v}+\frac{5(e^2 - 1)L_{uv}^2c^2}{4\tau^2L_v^2L_u}
\end{align}
in order to simplify the representation of the coefficients.
And to standardize terms on the left side of the inequality, we have defined $\lambda_1 =\frac{c}{4L_u}-\frac{15(e^2 - 1)c^3}{4\tau^2L_u}-\frac{15(e^2 - 1)L_{vu}^2c^2}{4\tau^2L_u^2L_v} $ and $\lambda_2 =\frac{c}{4L_v}-\frac{15(e^2 - 1)c^3}{4\tau^2L_v}-\frac{5(e^2 - 1)L_{uv}^2c^2}{4\tau^2L_v^2L_u} $.
\end{theorem}

\begin{remark}
\label{remark: convergence remark}
When dealing with more generalized loss functions that are non-convex and smooth, it is obvious that the loss associated with both global and personalized parameters is mainly influenced by two factors: the errors accumulated during the training phase and the variations introduced by data heterogeneity in the global aggregation process.
As training progresses through successive epochs, the gap between the initial loss and the optimal solution diminishes, asymptotically approaching zero. 
Provided that the learning rate is configured according to a constant coefficient scheme, we can express the optimal gap as: 
\begin{align}
    \mathcal{O}\left(\delta^2+\sigma_{u}^2+\sigma^2_{v}\right)
\end{align}

By linking the learning rate $\eta$ to the number of training epochs through $c = \mathcal{O}(T^{q}), q\in(-1,0)$ and put $c$ into $\eta$, we achieve more refined control over the global and personalized errors. Under this condition, we derive the following result:
\begin{align}
    \mathcal{O}\left(T^{1-q}\right), q \in (-1,0) 
\end{align}
\end{remark}

\section{Experiment Results}
In this section, we conduct several experiments to assess the performance of the algorithms discussed in previous sections under various parameter settings. 
We begin to introduce the configurations used to run these experiments.
Then we state the experimental results and evaluate the performance of our algorithm. 
In all experiments, we set the global training rounds for federated learning as \( T = 300 \), the local training rounds as \( \tau = 5 \), and the global initial learning rate as \( \eta = 1 \times 10^{-4} \). The local learning rates, whether \( \eta_u \) or \( \eta_v \), are set to \( \frac{\eta}{\sqrt{t}} \), with a step size of \( \frac{1}{k} \).

\subsection{Implementation Details}
{\bf Model and Dataset:}
The Swin Transformer Block serves as the core module across all algorithms in our implementation. 
While a four-stage architecture is commonly employed for multi-scale feature extraction in encoder or decoder, with multiple Swin Transformer Blocks typically deployed in each stage.
To accommodate hardware constraints, we have adapted the design by utilizing only two Swin Transformer Blocks per stage and operating with a reduced architecture of two selected stages across all algorithms.
The channel design in each algorithm follows a straightforward approach, where we assume the information is transmitted across an Additive White Gaussian Noise (AWGN) channel, SNR following a Gaussian distribution.
We train and evaluate PFL-DPJSCCA on image datasets with different resolutions from 32 × 32 up to 2K. 
Specifically, the CIFAR10 and DIV2K datasets are utilized for training, while the CIFAR10 and Kodak datasets are employed for testing. 
Notably, during training, high-resolution images are randomly cropped into 128 × 128 patches in each iteration.
Additionally, for low resolution images, we perform calculations based on a compression ratio of \( r = \frac{1}{6} \), while for high resolution images, we use a compression ratio of \( r = \frac{1}{16} \).

{\bf Baselines and Metrics:}
To empirically highlight the training performance of our proposed framework, we first compare it in the general settings with WITT \cite{10094735} in Fedavg \cite{pmlr-v54-mcmahan17a}, ADJSCC \cite{9438648} in FedLol \cite{10531097} and ADJSCC in Moon \cite{Li_2021_CVPR} based on CIFAR-10 and DIV2K datasets.
Then in order to evaluate the performance of different SNR distribution, we compare PFL-DPJSCCA in the similar settings with the above algorithms we mentioned. 
Finally, we conduct ablation experiments on PFL-DPJSCCA to meticulously analyze the roles and contributions of different modules within the framework.
Since we formulate image reconstruction as our task, we employ Peak Signal to Noise Ratio (PSNR) and MS-SSIM \cite{1292216} as evaluation metrics.

\subsection{Result Analysis}

\begin{table}[htbp]
\centering
\caption{Performance Comparison of PFL-DPJSCCA and Other Algorithms on DIV2K and CIFAR-10 under Various CSI}
\resizebox{\columnwidth}{!}{%
\begin{tabular}{cccccccccc}
\hline
\multirow{2}{*}{Dataset} & \multirow{2}{*}{\diagbox{ Distribution}{Algorithms}} & \multicolumn{2}{c}{WITT in FedAvg} & \multicolumn{2}{c}{ADJSCC in FedLol} & \multicolumn{2}{c}{ADJSCC in Moon} & \multicolumn{2}{c}{PFL-DPJSCCA} \\ \cline{3-10} 
                         &                   & PSNR            & MS-SSIM          & PSNR            & MS-SSIM            & PSNR           & MS-SSIM           & PSNR          & MS-SSIM         \\ \hline
\multirow{2}{*}{DIV2K}   & Mean=7.5          & 20.087          & 0.827            & 14.596          & 0.785              & 14.594          & 0.784            & \textbf{20.095}         & \textbf{0.829}          \\ \cline{2-10} 
                         & Mean=10.0         & 20.166          & 0.837            & 14.619          & 0.794              & 14.616          & 0.795            & \textbf{20.185}         & \textbf{0.840}          \\ \hline
\multirow{2}{*}{CIFAR10} & Mean=7.5          & 27.969          & 0.955            & 28.306          & 0.960              & 28.440          & 0.960            & \textbf{28.587}         & \textbf{0.961}          \\ \cline{2-10} 
                         & Mean=10.0         & 29.281          & 0.964            & \textbf{29.859}          & \textbf{0.970}              & 29.751          & 0.969            & 29.698         & 0.968          \\ \hline
\end{tabular}%
}
\label{tab: accuracy with baselines}
\end{table}

To evaluate the performance and generalization of PFL-DPJSCCA across datasets in comparison with baseline algorithms under different SNR conditions, we present Table \ref{tab: accuracy with baselines}.
In Table \ref{tab: accuracy with baselines}, it is evident that the baseline algorithms experience a notable performance decline as the channel conditions worsen. 
Under high SNR conditions, these baseline algorithms excel in extracting feature information and achieving better results. 
In contrast, PFL-DPJSCCA demonstrates remarkable stability and consistently superior performance across diverse datasets and various CSI scenarios.
Experimental results reveal that integrating the channel-aware mechanism into the feature extraction process delivers excellent performance for low-compression-rate image reconstruction tasks on low-resolution images. 
However, its effectiveness diminishes as tasks become more complex. Conversely, completely decoupling the channel-aware module from the feature extraction module leads to suboptimal performance on simpler tasks.
\begin{table}[htbp]
\centering
\caption{PFL-DPJSCCA Ablation Study}
\resizebox{\columnwidth}{!}{%
\begin{tabular}{@{}ccc|c|c|@{}}
\toprule
Personalized Federated Learning & Dual-Pipeline Encoder & Decoder Preprocess  & PSNR  &MS-SSIM          \\ \midrule
                                &                       &                     &20.063 &0.825            \\ \midrule
                                &\checkmark             & \checkmark          &20.077 &0.828            \\ \midrule
\checkmark                      &\checkmark             &                     &13.012 &0.714            \\ \midrule
\checkmark                      &\checkmark             & \checkmark          &20.095 &0.829            \\ \bottomrule
\end{tabular}%
}
\label{tab: accuracy with ablation}
\end{table}
Ablation experiments in Table \ref{tab: accuracy with ablation} show that dual-pipeline encoder and decoder preprocess significantly aids task reconstruction, but optimal performance is achieved only when combined with PFL. Both modules individually contribute positively.

\section{Conclusion}
\label{s:conclusion}

In this paper, we propose PFL-DPJSCCA. 
The dual pipelines and communication mechanism improve channel awareness and feature extraction, while the framework balances global collaboration and local personalization for better global model average performance. 
Extensive experiments under complex channel conditions and diverse datasets validate the model's performance, generalization ability, and stability in harsh environments. 
Ablation studies further confirm the positive impact of each designed module.

\bibliography{Reference}
\bibliographystyle{IEEEtran}


\clearpage
\begin{appendices}

\section{Assumptions and Theoretical Results}

To ensure the coherence of our proof process and facilitate the ease of writing, we will first review and summarize the appropriate assumptions that we have proposed. Following this, we will proceed to provide the proof of our theory.

\subsection{Review of Notations and Assumptions}
Firstly, we present some notations to elaborate on the details of the FL framework as follows.
We first consider a FL system with $N$ clients.
For any participating client, say client $n$, it has a local dataset $\mathcal{D}_n$ with $|\mathcal{D}_n|$ elements. 
Each element of $\mathcal{D}_n$ is a ground-true label $\bm{z}_{n} = \{\bm{x}_{n}, y_{n}\}$.
The $\bm{x}_{n} \in \mathcal{R}^{\text{in}}$ is the input vector and $\bm{y}_{n} \in \mathcal{R}^{\text{out}}$ is the output vector.
By utilizing the dataset $\mathcal{D}_n$ for $n=1,...,N$,  we aim to train $d_{u}$-dimension vector $\bm{u}$ and $d_{v}$-dimension vector $\bm{v}$ by minimizing the following loss function
\begin{small}
    \begin{align}
    F(\bm{u}, V) \triangleq \frac{1}{N} \sum_{n=1}^{N} F_n\left(\bm{u}_n,\bm{v}_n\right)=\frac{1}{N} \sum_{n=1}^{N} \frac{1}{D_{N}}\sum_{\bm{z}_{n}\in \mathcal{D}_n} f\left(\bm{u}_n,\bm{v}_n,\bm{z}_{n}\right)
\end{align}
\end{small}
where parameter $V$ is defined as the set of $\bm{v}_{n}, n=1,...,N$, and we denote local loss function is $f\left(\bm{u}_n,\bm{v}_n,\bm{d}_n\right)$.
For any client $n$, the personal parameters $\bm{v}_n$ and the local parameters $u_n$ update simultaneously, with their gradients are evaluated at the same time.

Secondly, we restate a set of general assumptions that are validated and utilized throughout the course of this work.
\begin{assumption}
\label{ass:smooth}
($L$-smooth).The loss function $F_n(\bm{u}_n,\bm{v}_n)$ is continuously differentiable for each client $n=1, ..., N$.
There exists constants $L>0$ such that for $u_n,\bm{v}_n, n=1,...,N$, the following inequality holds:
\begin{small}\begin{align}
    \nonumber
    F_n\left(\bm{u}_n^{t},\bm{v}_n\right) 
    &\le  F_n\left(\bm{u}_n^{t+1},\bm{v}_n\right) + \frac{L}{2}\|\bm{u}_n^{t+1}-\bm{u}_n^{t}\|^2 \\
    & +\langle \nabla F_n\left(\bm{u}_n^{t+1},\bm{v}_n\right),\bm{u}_n^{t+1}-\bm{u}_n^{t}\rangle
\end{align}\end{small}
Then the following inequality relations hold for each $F_n$:
\begin{small}\begin{align}
    \|F_n\left(\bm{u}_n^{t+1},\bm{v}_n\right) - F_n\left(\bm{u}_n^{t},\bm{v}_n\right)\|&\le L\|\bm{u}_n^{t+1}-\bm{u}_n^{t}\|
\end{align}\end{small}
\textup{Since the loss function $F_n$ involves two parameters $u$ and $v$, this necessitates the definition of four constants to fully characterize the smoothness properties of $F_n$ with respect to 
$u$ and $v$.
Let $L_u> 0$, $L_v> 0$, $L_{uv}> 0$, and $L_{vu}> 0$. Given $\nabla_{u}F_n\left(\bm{u}_n,\bm{v}_n\right)$ and $\nabla_{v}F_n\left(\bm{u}_n,\bm{v}_n\right)$, we define these positive constants to represent the following smoothness relations:
}
\begin{enumerate}[label=(\alph*)]
    \item $\nabla_{u}F_n\left(\bm{u}_n,\bm{v}_n\right)$ is $L_u$-smooth with respect to $u$,
    \item $\nabla_{v}F_n\left(\bm{u}_n,\bm{v}_n\right)$ is $L_v$-smooth with respect to $v$,
    \item $\nabla_{u}F_n\left(\bm{u}_n,\bm{v}_n\right)$ is $L_{uv}$-smooth with respect to $v$,
    \item $\nabla_{v}F_n\left(\bm{u}_n,\bm{v}_n\right)$ is $L_{vu}$-smooth with respect to $u$,
\end{enumerate}
\end{assumption}

\begin{assumption}
\label{ass: Bounded Inner Error}
(Bounded Inner Error).Suppose a subset of dataset $\mathcal{D}_{n}$ is selected randomly, which is denote as $\xi_{n}$, define
\begin{equation}
    F_{n}\left(\bm{u}_n,\bm{v}_n,\xi_n\right) \triangleq \frac{1}{|\xi_n|} \sum_{\bm{z}_n\in \xi_n}f\left(\bm{u}_n,\bm{v}_n,\bm{z}_n\right)
\end{equation}
We assume that the stochastic gradients $G$ in this work are unbiased, i.e.,
\begin{small}\begin{align}
    &\nabla_{u} F_{n}\left(\bm{u}_n,\bm{v}_n\right)=\mathbf{E}\left[ G_{n,u}\left(\bm{u}_n,\bm{v}_n,\xi_n\right)\right] \\
    &\nabla_{v} F_{n}\left(\bm{u}_n,\bm{v}_n\right)=\mathbf{E}\left[ G_{n,v}\left(\bm{u}_n,\bm{v}_n,\xi_n\right)\right]
\end{align}\end{small}
Furthermore, for each client $n=1, ..., N$, with $\xi_{n} \in \mathcal{D}_n$, when considering $\nabla F_n\left(\bm{u}_n,\bm{v}_n\right)$ and $G_n\left(\bm{u}_n,\bm{v}_n,\xi_{n}\right)$, the inner error of stochastic gradients is assumed to be bounded. There exist constant $\sigma_{u}$, $\sigma_{v}$, such that
\begin{small}\begin{align} 
\|G_{n,u}\left(\bm{u}_n,\bm{v}_n,\xi_{n}\right)-\nabla_{u}F_n\left(\bm{u}_n,\bm{v}_n\right)\|^2&\le  \sigma^2_{u} \\
\|G_{n,v}\left(\bm{u}_n,\bm{v}_n,\xi_{n}\right)-\nabla_{v}F_n\left(\bm{u}_n,\bm{v}_n\right)\|^2&\le  \sigma^2_{v} 
\end{align}\end{small}

\textup{For each client, this assumption is commonly adopted to bound the inner error resulting from the random selection of data for the current training session \cite{huang2023achieving,pmlr-v202-li23o}. 
During the local training phase, the datasets across clients exhibit non-iid characteristics, which is referred to as data heterogeneity. 
The outer error in the loss function, induced by such data heterogeneity, will be imposed in the next assumption.}
\end{assumption}

\begin{assumption}
\label{ass: Bounded Outer Error}
(Bounded Outer Error).For $F_n\left(\bm{u}_n,\bm{v}_n\right)$, the variance of stochastic gradients is assumed to be bounded by utilizing constants $\rho, \delta$, i.e.,
\begin{small}\begin{align}
&\frac{1}{N} \sum_{n=1}^{N}\|\nabla_{u} F_n\left(\bm{u}_n,\bm{v}_n\right) - \nabla_{u} F\left(\bm{u}_n,V\right)\|^2  \\
\nonumber
&\qquad\qquad\qquad\qquad\qquad\qquad\le \rho^2 \|\nabla_{u} F_n\left(\bm{u}_n,V\right)\|^2 +\delta^2
\end{align}\end{small}
\textup{Generally, the assumption in the special case of $\rho=0$ within this framework has been extensively put forward \cite{pmlr-v119-koloskova20a}. 
This inequality serves to characterize the degree of detriment to global convergence inflicted by each client as a consequence of the data heterogeneity issue. 
Specifically, this beneficial or detrimental effect varies in accordance with the fluctuations of the global loss function gradient within the scope defined by $\rho$}
\end{assumption}

\subsection{Convergence Analysis and Proof of Theorem}

We define $F(\bm{u}^{t}, V^{t})$ is the global loss function, which $V$ represents the concatenation of personalized parameters owned by all clients.
First, we consider $F(\bm{u}^{t + 1}, V^{t+1})-F(\bm{u}^{t}, V^{t})$, it can be decomposed into the following formulas:
\begin{small}\begin{align}
    \nonumber
    &F(\bm{u}^{t + 1}, V^{t+1})-F(\bm{u}^{t}, V^{t}) \\
    \nonumber
    & =\frac{1}{N}\sum_{n = 1}^{N}F_n(\bm{u}^{t + 1},\bm{v}_n^{t+1})-\frac{1}{N}\sum_{i = 1}^{N}F_n(\bm{u}^{t},\bm{v}_n^{t}) \\
    &=\frac{1}{N}\sum_{n = 1}^{N}[F_n(\bm{u}^{t + 1},\bm{v}_n^{t+1})-F_n(\bm{u}^{t},\bm{v}_n^{t})]
\end{align}\end{small}

By combining the formulas above and Lemma \ref{e: lemma1}, we obtain:
\begin{small}\begin{align}
\label{e: four parts}
\nonumber
&\mathbb{E}[F(\bm{u}^{t + 1},V^{t + 1})]-F(\bm{u}^{t},V^{t})\\
\nonumber
&\leq\mathbf{E}\langle\nabla_{u}F(\bm{u}^{t},V^{t}),\bm{u}^{t + 1}-\bm{u}^{t}\rangle+\frac{L_u + L_{uv}}{2}\mathbf{E}\|\bm{u}^{t + 1}-\bm{u}^{t}\|^2 \\
&+\frac{1}{N}\sum_{n = 1}^{N}\mathbf{E}\langle\nabla_{v}F_n(\bm{u}^{t},\bm{v}_n^{t}),\bm{v}_n^{t + 1}-\bm{v}_n^{t}\rangle+\frac{L_v+ L_{uv}}{2N}\sum_{n = 1}^{N}\mathbf{E}\|\bm{v}_n^{t + 1}-\bm{v}_n^{t}\|^2
\end{align}\end{small}

Then, plugging the results of Lemma \ref{e: lemma2} to  \ref{e: lemma5}, we can get the bound:
\begin{small}\begin{align}
    \nonumber
    &\mathbb{E}[F(\bm{u}^{t + 1}, V^{t + 1})] - F(\bm{u}^{t}, V^{t}) \\
    &\leq \left(\frac{3\left(L_{v}+L_{uv}\right)\eta_{v}^{2}\tau^{2}}{2N}-\frac{\eta\gamma}{2N}\right)\sum_{n = 1}^{N}\left\|\nabla_{v}F_{n}(\bm{u}^{t},\bm{v}_{n}^{t})\right\|^{2} \\
    \nonumber
    &+\left(\frac{3\left(L_{u}+L_{uv}\right)\eta_{u}^{2}\tau^{2}(1+\rho^{2})}{2}-\frac{\eta_{u}\tau}{2}\right)\left\|\nabla_{u}F_{n}(\bm{u}^{t},V^{t})\right\|^{2} \\
    \nonumber
    &+\left(\frac{\eta_{u}L_{u}^{2}}{N}+\frac{3\left(L_{u}+L_{uv}\right)L_{u}^{2}\eta_{u}^{2}\tau}{2N}\right)\sum_{n = 1}^{N}\sum_{k = 0}^{\tau - 1}\mathbf{E}\|\bm{u}_{n,k}^{t}-\bm{u}^{t}\|^{2} \\
    \nonumber
    &+\left(\frac{\eta_{u}L_{uv}^{2}}{N}+\frac{3\left(L_{u}+L_{uv}\right)L_{uv}^{2}\eta_{u}^{2}\tau}{2N}\right)\sum_{n = 1}^{N}\sum_{k = 0}^{\tau - 1}\mathbf{E}\|\bm{v}_{n,k}^{t}-\bm{v}_n^{t}\|^{2} \\
    \nonumber
    &+\left(\frac{\eta_{v} L_{vu}^2}{N}+\frac{3\left(L_v + L_{uv}\right)L_{vu}^2\eta_v^2\tau}{2N}\right)\sum_{n = 1}^{N}\sum_{k = 0}^{T - 1}\mathbf{E}\|\bm{u}_{n,k}^t - \bm{u}_n^t\|^2\\
    \nonumber
    &+\left(\frac{\eta_{v} L_{v}^2}{N}+\frac{3\left(L_v + L_{uv}\right)L_{v}^2\eta_v^2\tau}{2N}\right)\sum_{n = 1}^{N}\sum_{k = 0}^{T - 1}\mathbf{E}\|\bm{v}_{n,k}^t - \bm{v}_n^t\|^2\\
    \nonumber
    &+\frac{\left(L_{u}+L_{uv}\right)\eta_{u}^{2}\tau^{2}}{2}\sigma_{u}^{2}+\frac{3\left(L_{u}+L_{uv}\right)\eta_{u}^{2}\tau^{2}}{2}\delta^{2} \\
    \nonumber
    &+\frac{\left(L_{v}+L_{uv}\right)\eta_{v}^{2}\tau^2}{2}\sigma_{v}^{2}
\end{align}\end{small}

For convenience in obtaining subsequent proofs, we assume that $\eta_{u} \le \frac{1}{6\tau\left(\max\left\{L_{u},L_{v}\right\}+L_{uv}\right)\left(1+\rho^2\right)}$ and $\eta_{v} \le \frac{1}{6\tau\left(\max\left\{L_{u},L_{v}\right\}+L_{uv}\right)}$ to simplify the constants in inequality. So we have 
\begin{small}\begin{align}
    \nonumber
    &\mathbb{E}[F(\bm{u}^{t + 1}, V^{t + 1})] - F(\bm{u}^{t}, V^{t}) \\
    \nonumber
    &\leq -\frac{\eta_{u} \tau}{4}\left\|\nabla_{u} F\left(\bm{u}^{t}, V^{t}\right)\right\|^{2}-\frac{\eta_{v} \tau}{4n} \sum_{n = 1}^{N}\left\|\nabla_{v} F_{n}\left(\bm{u}_{n}^{t}, \bm{v}_{n}^{t}\right)\right\|^{2} \\
    \nonumber
    &+\frac{5}{4 N}\left(L_{u}^{2}\eta_{u}+L_{v u}^{2}\eta_{v}\right) \sum_{n = 1}^{N} \sum_{k = 0}^{\tau}\left\|\bm{u}_{n, k}^{t}-\bm{u}^{t}\right\|^{2} \\
    \nonumber
    &+\frac{5}{4 N}\left(L_{v}^{2}\eta_{v}+L_{u v}^{2}\eta_{u}\right) \sum_{n = 1}^{N} \sum_{k = 0}^{\tau}\left\|\bm{v}_{n, k}^{t}-\bm{v}_{n}^{t}\right\|^{2} \\
    \nonumber
    &+\frac{\left(L_{u}+L_{u v}\right) \tau^2\eta_{u}^2}{2} \sigma_{u}^{2}+\frac{\left(L_{v}+L_{u v}\right) \tau^{2}\eta_{v}^2}{2} \sigma_{v}^{2} \\
    &+\frac{3\left(L_{u}+L_{u v}\right) \tau^{2}\eta_{u}^2}{2}\delta^2 
\end{align}\end{small}

Based on the Lemma \ref{e: lemma6}, we can integrate the above inequality to obtain:
\begin{small}
\begin{align}
    \nonumber
    &\mathbb{E}[F(\bm{u}^{t + 1}, V^{t + 1})] - F(\bm{u}^{t}, V^{t}) \\
    \nonumber
    &\le\left(\frac{15(e^2-1)\tau\eta_u^{2}}{4}\left(L_{u}^{2}\eta_{u}+L_{vu}^2\eta_v\right) -\frac{\eta_{u} \tau}{4} \right)\left\|\nabla_{u} F\left(\bm{u}^{t}, V^{t}\right)\right\|^{2} \\
    \nonumber
    &+\left(\frac{15(e^2-1)\tau\eta_{v}^2}{4N}\left(L_{v}^{2}\eta_{v}+L_{uv}^2\eta_u\right)-\frac{\eta_{v} \tau}{4N}\right)\sum_{n = 1}^{N}\left\|\nabla_{v} F_{n}\left(\bm{u}_{n}^{t}, \bm{v}_{n}^{t}\right)\right\|^{2} \\
    \nonumber
    &+\left(\frac{\left(L_{u}+L_{u v}\right) \tau^2\eta_{u}^2}{2} +\frac{5(e^2-1)}{4}\left(L_{u}^{2}\eta_{u}+L_{vu}^2\eta_v\right) \right)\sigma_{u}^{2} \\
    \nonumber
    &+\left(\frac{\left(L_{v}+L_{u v}\right) \tau^{2}\eta_{v}^2}{2} +\frac{5(e^2-1)}{4}\left(L_{v}^{2}\eta_{v}+L_{uv}^2\eta_u\right) \right)\sigma_{v}^{2} \\
    \label{e: gap fo near round}
    &+\frac{3\left(L_{u}+L_{u v}\right) \tau^{2}\eta_{u}^2}{2}\delta^2
\end{align}
\end{small}

Substitute  $\eta_{u} \le \frac{1}{\sqrt{6}\tau L_{u}\max\left\{1,\frac{L_{vu}L_{uv}}{L_{u}L_{v}}\right\}}$ and $\eta_{v} \le \frac{1}{\sqrt{6}\tau L_{v}\max\left\{1,\frac{L_{vu}L_{uv}}{L_{u}L_{v}}\right\}}$ into the inequality (\ref{e: gap fo near round}). 
That is, when a reasonable learning rate is set, we can obtain the following results
\begin{small}
\begin{align}
\nonumber
&\mathbb{E}[F(\bm{u}^{t + 1}, V^{t + 1})] - F(\bm{u}^{t}, V^{t}) \\
\nonumber
&\leq \left(\frac{15(e^2 - 1)c^3}{4\tau^2L_u}+\frac{15(e^2 - 1)L_{vu}^2c^2}{4\tau^2L_uL_v}-\frac{c}{4L_u}\right)\left\|\nabla_u F(\bm{u}^{t}, V^{t})\right\|^2 \\
\nonumber
&+\left(\frac{15(e^2 - 1)c^3}{4\tau^2L_v}+\frac{15(e^2 - 1)L_{uv}^2c^2}{4\tau^2L_uL_v^2}-\frac{c}{4L_v}\right)\frac{1}{N}\sum_{n = 1}^{N}\left\|\nabla_{v_n}F(\bm{u}_n^{t}, \bm{v}_n^{t})\right\|^2\\
\nonumber
&+\left[+\frac{(L_u + L_{uv})c^2}{2L_u}+\frac{5(e^2 - 1)c^3}{4\tau^2L_u}+\frac{5(e^2 - 1)L_{vu}^2c^2}{4\tau^2L_u^2L_v}\right]\sigma^2_{u} \\
\nonumber
&+\left[\frac{(L_v + L_{uv})c^2}{2L_v}+\frac{5(e^2 - 1)c^3}{4\tau^2L_v}+\frac{5(e^2 - 1)L_{uv}^2c^2}{4\tau^2L_v^2L_u}\right]\sigma^2_{v} \\
&+\frac{3(L_u + L_{uv})c^2}{2L_u^2}\delta^2
\end{align}
\end{small}
where we denote $c$ that is the same uncertain range of learning rate for $\bm{u}$ and $\bm{v}$.

Telescoping $T$ rounds of iterations , it follows that
\begin{small}
\begin{align}
\nonumber
&\mathbb{E}[F(\bm{u}^{T}, V^{T})] - F(\bm{u}^{0}, V^{0}) \\
\nonumber
&\leq \left(\frac{15(e^2 - 1)c^3}{4\tau^2L_u}+\frac{15(e^2 - 1)L_{vu}^2c^2}{4\tau^2L_uL_v}-\frac{c}{4L_u}\right)\sum_{t=0}^{T-1}\left\|\nabla_u F(\bm{u}^{t}, V^{t})\right\|^2 \\
\nonumber
&+\left(\frac{15(e^2 - 1)c^3}{4\tau^2L_v}+\frac{15(e^2 - 1)L_{uv}^2c^2}{4\tau^2L_uL_v^2}-\frac{c}{4L_v}\right)\frac{1}{N}\sum_{t=0}^{T-1}\sum_{n = 1}^{N}\left\|\nabla_{v}F(\bm{u}_n^{t}, \bm{v}_n^{t})\right\|^2\\
\nonumber
&+\left[\frac{(L_u + L_{uv})c^2}{2L_u}+\frac{5(e^2 - 1)c^3}{4\tau^2L_u}+\frac{5(e^2 - 1)L_{vu}^2c^2}{4\tau^2L_u^2L_v}\right]T\sigma^2_{u} \\
\nonumber
&+\left[\frac{(L_v + L_{uv})c^2}{2L_v}+\frac{5(e^2 - 1)c^3}{4\tau^2L_v}+\frac{5(e^2 - 1)L_{uv}^2c^2}{4\tau^2L_v^2L_u}\right]T\sigma^2_{v} \\
&+\frac{3(L_u + L_{uv})c^2}{2L_u^2}T\delta^2
\end{align}
\end{small}
We assume that $F^*$ represents the optimal value that the global function can achieve under the conditions of the optimal global shared parameters and the best-match local personalized parameters.
For $F^*$, we have the following properties:
\begin{small}\begin{align}
    F^* \le F(\bm{u},\bm{v}), & \forall u\in\mathcal{R}^{d_{u}^{in}} , \forall v \in \mathcal{R}^{d_{v}^{in}}\\
    &\nabla F^* =0
\end{align}\end{small}

So we have
\begin{small}\begin{align}
    F^*-F(\bm{u}^{0}, V^{0}) \le F(\bm{u}^{T}, V^{T})-F(\bm{u}^{0}, V^{0})
\end{align}\end{small}

When we assume that the coefficients in $\left\|\nabla_u F(\bm{u}^{t}, V^{t})\right\|^2$ and $\left\|\nabla_{v}F(\bm{u}_n^{t}, \bm{v}_n^{t})\right\|^2$ are negative, by moving the terms to the left hand of the inequality and also relocating the terms originally on the right hand of the inequality to the left, we obtain:
\begin{small}
\begin{align}
\nonumber
&\left[\frac{c}{4L_u}-\frac{15(e^2 - 1)c^3}{4\tau^2L_u}-\frac{15(e^2 - 1)L_{vu}^2c^2}{4\tau^2L_u^2L_v}\right]\frac{1}{T}\sum_{t = 0}^{T}\left\|\nabla_u F(\bm{u}^{t}, V^{t})\right\|^2 \\
\nonumber
&+\left[\frac{c}{4L_v}-\frac{15(e^2 - 1)c^3}{4\tau^2L_v}-\frac{5(e^2 - 1)L_{uv}^2c^2}{4\tau^2L_v^2L_u}\right]\frac{1}{NT}\sum_{t = 0}^{T}\sum_{n = 1}^{N}\left\|\nabla_{v_n} F(\bm{u}_n^{t}, \bm{v}_n^{t})\right\|^2\\
\nonumber
&\leq\frac{F^* - F(\bm{u}^0, V^0)}{T}+\frac{3(L_u+L_{uv})c^2}{2L_u^2}\delta^2\\
\nonumber
&+\left[\frac{(L_u + L_{uv})c^2}{2L_u^2}+\frac{5(e^2 - 1)c^3}{4\tau^2L_u}+\frac{5(e^2 - 1)L_{vu}^2c^2}{4\tau^2L_u^2L_v}\right]\sigma^2_{u} \\
&+\left[\frac{(L_v + L_{uv})c^2}{2L_v^2}+\frac{5(e^2 - 1)c^3}{4\tau^2L_v}+\frac{5(e^2 - 1)L_{uv}^2c^2}{4\tau^2L_v^2L_u}\right]\sigma^2_{v}
\end{align}
\end{small}

By further simplifying the inequality, we can obtain the final result when we define $\lambda_1 =\frac{c}{4L_u}-\frac{15(e^2 - 1)c^3}{4\tau^2L_u}-\frac{15(e^2 - 1)L_{vu}^2c^2}{4\tau^2L_u^2L_v} $ and $\lambda_2 =\frac{c}{4L_v}-\frac{15(e^2 - 1)c^3}{4\tau^2L_v}-\frac{5(e^2 - 1)L_{uv}^2c^2}{4\tau^2L_v^2L_u} $
\begin{small}
\begin{align}
\nonumber
&\frac{1}{T}\sum_{t = 0}^{T}\left\|\nabla_{u}F(\bm{u}^{t},V^{t})\right\|^2+\frac{1}{NT}\sum_{t = 0}^{T}\sum_{n = 1}^{N}\left\|\nabla_{v}F(\bm{u}_{n}^{t},\bm{v}_{n}^{t})\right\|^2\\
\nonumber
&\leq\frac{\lambda_1}{\min\{\lambda_1,\lambda_2\}}\frac{1}{T}\sum_{t = 0}^{T}\left\|\nabla_{u}F(u^{t},V^{t})\right\|^2 \\
\nonumber
&\qquad\qquad+\frac{\lambda_2}{\min\{\lambda_1,\lambda_2\}}\frac{1}{NT}\sum_{t = 0}^{T}\sum_{n = 1}^{N}\left\|\nabla_{v_n}F(u_{n}^{t},v_{n}^{t})\right\|^2\\
&\leq\frac{F^* - F(\bm{u}^0,V^0)}{T\min\{\lambda_1,\lambda_2\}}+\frac{3(L_u + L_{uv})c^2}{2L_u^2\min\{\lambda_1,\lambda_2\}}\delta^2\\
\nonumber
&\qquad\qquad+\frac{\left[\frac{(L_u + L_{uv})c^2}{2L_u^2}+\frac{5(e^2 - 1)c^3}{4\tau^2L_u}+\frac{5(e^2 - 1)L_{vu}^2c^2}{4\tau^2L_u^2L_v}\right]}{\min\{\lambda_1,\lambda_2\}}\sigma_{u}^2 \\
\nonumber
&\qquad\qquad+\frac{\left[\frac{(L_v + L_{uv})c^2}{2L_v^2}+\frac{5(e^2 - 1)c^3}{4\tau^2L_v}+\frac{5(e^2 - 1)L_{uv}^2c^2}{4\tau^2L_v^2L_u}\right]}{\min\{\lambda_1,\lambda_2\}}\sigma_{v}^2
\end{align}
\end{small}


\begin{lemma}
\label{e: lemma1}
According to our assumptions, we bound the loss function gap for any client from round $t$ to round $t+1$:
\begin{small}\begin{align}
&F_n(\bm{u}^{t + 1},\bm{v}_n^{t+1})-F_n(\bm{u}^{t},\bm{v}_n^{t}) \\
\nonumber
&\leq\langle\nabla_{u}F_n(\bm{u}^{t},\bm{v}_n^{t}),\bm{u}^{t + 1}-\bm{u}^{t}\rangle+\frac{L_u + L_{uv}}{2}\|\bm{u}^{t + 1}-\bm{u}^{t}\|^2 \\
\nonumber
&+\langle\nabla_{v}F_n(\bm{u}^{t},\bm{v}_n^{t}),\bm{v}_n^{t + 1}-\bm{v}_n^{t}\rangle+\frac{L_v+ L_{uv}}{2}\|\bm{v}_n^{t + 1}-\bm{v}_n^{t}\|^2
\end{align}\end{small}
\end{lemma}

\begin{proof}
Based on the assumptions, we sequentially apply the L-smooth inequality to $u$ and $v$ to facilitate the completion of the proof.
\begin{subequations}
\begin{small}\begin{align}
\nonumber
&F_n(\bm{u}^{t + 1},\bm{v}_n^{t+1})-F_n(\bm{u}^{t},\bm{v}_n^{t}) \\
&\leq\langle\nabla_{u}F_n(\bm{u}^{t},\bm{v}_n^{t+1}),\bm{u}^{t + 1}-\bm{u}^{t}\rangle +\frac{L_u}{2}\|\bm{u}^{t + 1}-\bm{u}^{t}\|^2\\
\nonumber
&\qquad+\langle\nabla_{v}F_n(\bm{u}^{t},\bm{v}_n^{t}),\bm{v}_n^{t + 1}-\bm{v}_n^{t}\rangle+\frac{L_v}{2}\|\bm{v}_n^{t + 1}-\bm{v}_n^{t}\|^2\\
&\leq\langle\nabla_{u}F_n(\bm{u}^{t},\bm{v}_n^{t + 1})-\nabla_{u}F_n(\bm{u}^{t},\bm{v}_i^{t}),\bm{u}^{t + 1}-\bm{u}^{t}\rangle \\
\nonumber
&\qquad+\langle\nabla_{u}F_n(\bm{u}^{t},\bm{v}_n^{t}),\bm{u}^{t + 1}-\bm{u}^{t}\rangle +\frac{L_u}{2}\|\bm{u}^{t + 1}-\bm{u}^{t}\|^2\\
\nonumber
&\qquad+\langle\nabla_{v}F_n(\bm{u}^{t},\bm{v}_n^{t}),\bm{v}_n^{t + 1}-\bm{v}_n^{t}\rangle+\frac{L_v}{2}\|\bm{v}_n^{t + 1}-\bm{v}_n^{t}\|^2\\
\label{lemma1-1}
&\leq\|\nabla_{u}F_n(\bm{u}^{t},\bm{v}_n^{t + 1})-\nabla_{u}F_n(\bm{u}^{t},\bm{v}_n^{t})\|\|\bm{u}^{t + 1}-\bm{u}^{t}\| \\
\nonumber
&\qquad+\langle\nabla_{u}F_n(\bm{u}^{t},\bm{v}_n^{t}),\bm{u}^{t + 1}-\bm{u}^{t}\rangle+\frac{L_u}{2}\|\bm{u}^{t + 1}-\bm{u}^{t}\|^2 \\
\nonumber
&\qquad+\langle\nabla_{v}F_n(\bm{u}^{t},\bm{v}_n^{t}),\bm{v}_n^{t + 1}-\bm{v}_n^{t}\rangle+\frac{L_v}{2}\|\bm{v}_n^{t + 1}-\bm{v}_n^{t}\|^2\\
&\leq L_{uv}\|\bm{v}_n^{t + 1}-\bm{v}_n^{t}\|\|\bm{u}^{t + 1}-\bm{u}^{t}\| \\
\nonumber
&\qquad+\langle\nabla_{u}F_n(\bm{u}^{t},\bm{v}_n^{t}),\bm{u}^{t + 1}-\bm{u}^{t}\rangle+\frac{L_u}{2}\|\bm{u}^{t + 1}-\bm{u}^{t}\|^2 \\
\nonumber
&\qquad+\langle\nabla_{v}F_n(\bm{u}^{t},\bm{v}_n^{t}),\bm{v}_n^{t + 1}-\bm{v}_n^{t}\rangle+\frac{L_v}{2}\|\bm{v}_n^{t + 1}-\bm{v}_n^{t}\|^2\\
\label{lemma1-2}
&\leq\frac{L_{uv}}{2}\|\bm{v}_n^{t + 1}-\bm{v}_n^{t}\|^2+\frac{L_{uv}}{2}\|\bm{u}^{t + 1}-\bm{u}^{t}\|^2|\\
\nonumber
&\qquad+\langle\nabla_{u}F_n(\bm{u}^{t},\bm{v}_n^{t}),\bm{u}^{t + 1}-\bm{u}^{t}\rangle+\frac{L_u}{2}\|\bm{u}^{t + 1}-\bm{u}^{t}\|^2 \\
\nonumber
&\qquad+\langle\nabla_{v}F_n(\bm{u}^{t},\bm{v}_n^{t}),\bm{v}_n^{t + 1}-\bm{v}_n^{t}\rangle+\frac{L_v}{2}\|\bm{v}_n^{t + 1}-\bm{v}_n^{t}\|^2\\
\nonumber
&=\langle\nabla_{u}F_n(\bm{u}^{t},\bm{v}_n^{t}),\bm{u}^{t + 1}-\bm{u}^{t}\rangle+\frac{L_u + L_{uv}}{2}\|\bm{u}^{t + 1}-\bm{u}^{t}\|^2 \\
\nonumber
&+\langle\nabla_{v}F_n(\bm{u}^{t},\bm{v}_n^{t}),\bm{v}_n^{t + 1}-\bm{v}_n^{t}\rangle+\frac{L_v+ L_{uv}}{2}\|\bm{v}_n^{t + 1}-\bm{v}_n^{t}\|^2
\end{align}\end{small}
\end{subequations}
For (\ref{lemma1-1}) and (\ref{lemma1-2}), we employ the triangle inequality, $\bm{a}^T\bm{b} \leq \|\bm{a}\|\|\bm{b}\| \leq \frac{\|\bm{a}\|^2}{2}+\frac{\|\bm{b}\|^2}{2}$, thereby concluding the proof.
\end{proof}

\begin{lemma}
\label{e: lemma2}
We define the number of local epochs run by all clients is $\tau$. So we can bound $\langle\nabla_{u}F_n(\bm{u}^{t},V^{t}),\bm{u}^{t + 1}-\bm{u}^{t}\rangle$:
\begin{small}\begin{align}
\nonumber
&\mathbf{E}\left\{\langle\nabla_{u}F_n(\bm{u}^{t},V^{t}),\bm{u}^{t + 1}-\bm{u}^{t}\rangle\right\} \\
\nonumber
&\leq-\frac{\eta_{u}\tau}{2}\|\nabla_{u}F_n(\bm{u}^{t},V^{t})\|^{2} \\
&\qquad+\frac{\eta_{u}}{n}\sum_{n = 1}^{N}\sum_{k = 0}^{\tau - 1}\mathbf{E}\left[L_{u}^{2}\|\bm{u}_{n,k}^{t}-\bm{u}^{t}\|^{2}+L_{uv}^{2}\|\bm{v}_{n,k}^{t}-\bm{v}_n^{t}\|^{2}\right]
\end{align}\end{small}
\end{lemma}
\begin{proof}
Based on our assumptions, we take the expectation of the unbiased gradient for any client, thereby obtaining the following:
\begin{small}\begin{align}
\nonumber
&\mathbf{E}\left\{\langle\nabla_{u}F_n(\bm{u}^{t},V^{t}),\bm{u}^{t + 1}-\bm{u}^{t}\rangle\right\} \\
\nonumber
&=\eta_{u}\sum_{k = 0}^{\tau - 1}\mathbf{E}\left\{\langle\nabla_{u}F_n(\bm{u}^{t},V^{t}),\frac{1}{N}\sum_{n = 1}^{N}\nabla_{u}F_{n}(\bm{u}_{n,k}^{t},\bm{v}_{n,k}^{t})\rangle\right\} \\
\nonumber
&\leq-\frac{\eta_{u}\tau}{2}\|\nabla_{u}F_n(\bm{u}^{t},V^{t})\|^{2} \\
\nonumber
&\qquad+\frac{\eta_{u}}{2}\sum_{k = 0}^{\tau - 1}\mathbf{E}\|\frac{1}{N}\sum_{n = 1}^{N}\nabla_{u}F_{n}(\bm{u}_{n,k}^{t},\bm{v}_{n,k}^{t})-\nabla_{u}F_n(\bm{u}^{t},V^{t})\|^{2} \\
\nonumber
&\leq-\frac{\eta_{u}\tau}{2}\|\nabla_{u}F_n(\bm{u}^{t},V^{t})\|^{2} \\
\nonumber
&\qquad+\frac{\eta_{u}}{2N}\sum_{k = 0}^{\tau - 1}\sum_{n = 1}^{N}\mathbf{E}\|\nabla_{u}F_{n}(\bm{u}_{n,k}^{t},\bm{v}_{n,k}^{t})-\nabla_{u}F_n(\bm{u}^{t},\bm{v}_n^{t})\|^{2} \\
&\leq-\frac{\eta_{u}\tau}{2}\|\nabla_{u}F_n(\bm{u}^{t},V^{t})\|^{2} \\
\nonumber
&\qquad+\frac{\eta_{u}}{N}\sum_{n = 1}^{N}\sum_{k = 0}^{\tau - 1}\mathbf{E}\left[L_{u}^{2}\|\bm{u}_{n,k}^{t}-\bm{u}^{t}\|^{2}+L_{uv}^{2}\|\bm{v}_{n,k}^{t}-\bm{v}_n^{t}\|^{2}\right]
\end{align}\end{small}
\end{proof}

\begin{lemma}
\label{e: lemma3}
Following a similar approach as in Lemma \ref{e: lemma2}, bounds can be established for $\frac{1}{N}\sum_{n = 1}^{N}\mathbf{E}\langle\nabla_{v}F_n(\bm{u}^t,\bm{v}_n^t),\bm{v}_n^{t + 1}-\bm{v}_n^t\rangle$ as well.
\begin{small}\begin{align}
\nonumber
&\frac{1}{N}\sum_{n = 1}^{N}\mathbf{E}\langle\nabla_{v}F_n(\bm{u}^t,\bm{v}_n^t),\bm{v}_n^{t + 1}-\bm{v}_n^t\rangle \\
&\leq-\frac{\eta_{v}\gamma}{2N}\sum_{n = 1}^{N}\|\nabla_{v}F_n(\bm{u}^t,\bm{v}_n^t)\|^2 \\
\nonumber
&\qquad+\frac{\eta_{v}}{N}\sum_{n = 1}^{N}\sum_{k = 0}^{T - 1}\mathbf{E}\left[L_{vu}^{2}\|\bm{u}_{n,k}^t - \bm{u}_n^t\|^2+L_{v}^2\|\bm{v}_{n,k}^t - \bm{v}_n^t\|^2\right]
\end{align}\end{small}
\end{lemma}

\begin{proof}
To address the personalized model parameters for any client $n$, we have:
\begin{small}\begin{align}
\label{lemma3-1}
\nonumber
&\mathbf{E}\langle\nabla_{v}F_n(\bm{u}^t,\bm{v}_n^t),\bm{v}_n^{t + 1}-\bm{v}_n^t\rangle \\
\nonumber
&=-\eta_{v}\sum_{k = 0}^{\tau - 1}\mathbf{E}\langle\nabla_{v}F_n(\bm{u}^t,\bm{v}_n^t),\nabla_{v}F_n(\bm{u}_{n,k}^t,\bm{v}_{n,k}^t) \\
\nonumber
&=-\eta_{v}\tau\|\nabla_{v}F_n(\bm{u}^t,\bm{v}_n^t)\|^2 \\
\nonumber
&\quad-\eta_{v}\sum_{k = 0}^{\tau - 1}\mathbf{E}\langle\nabla_{v}F_n(\bm{u}^t,\bm{v}_n^t),\nabla_{v}F_n(\bm{u}_{n,k}^t,\bm{v}_{n,k}^t) -\nabla_{v}F_n(\bm{u}^t,\bm{v}_n^t)\rangle \\
\nonumber
&\leq\frac{\eta_{v}\tau}{2}\|\nabla_{v}F_n(\bm{u}^t,\bm{v}_n^t)\|^2 \\
&\quad+\frac{\eta_{v}}{2}\sum_{k = 0}^{\tau - 1}\mathbf{E}\|\nabla_{v}F_n(\bm{u}_{n,k}^t,\bm{v}_{n,k}^t)-\nabla_{v}F_n(\bm{u}^t,\bm{v}_n^t)\|^2
\end{align}\end{small}
The last inequality can be derived from $-\bm{a}^T\bm{b} \leq \frac{1}{2}\left(\|\bm{a}\|^2+\|\bm{b}\|^2\right)$.

Regarding the last term of the inequality, the following bound can be derived:
\begin{small}\begin{align}
\label{lemma3-2}
\nonumber
&\|\nabla_{v}F_n(\bm{u}_{n,k}^t,\bm{v}_{n,k}^t)-\nabla_{v}F_n(\bm{u}^t,\bm{v}_n^t)\|^2 \\
\nonumber
&=\|\nabla_{v}F_n(\bm{u}_{n,k}^t,\bm{v}_{n,k}^t)-\nabla_{v}F_n(\bm{u}^t,\bm{v}_{n,k}^t) \\
\nonumber
&\qquad\qquad\qquad\qquad+\nabla_{v}F_n(\bm{u}^t,\bm{v}_{n,k}^t)-\nabla_{v}F_n(\bm{u}^t,\bm{v}_n^t)\|^2 \\
\nonumber
&\leq2\|\nabla_{v}F_n(\bm{u}_{n,k}^t,\bm{v}_{n,k}^t)-\nabla_{v}F_n(\bm{u}^t,\bm{v}_{n,k}^t)\|^2 \\
\nonumber
&\qquad\qquad\qquad\qquad+ 2\|\nabla_{v}F_n(\bm{u}^t,\bm{v}_{n,k}^t)-\nabla_{v}F_n(\bm{u}^t,\bm{v}_n^t)\|^2 \\
& \leq2L_{vu}^2\|\bm{u}_{n,k}^t - \bm{u}^t\|^2+2L_{v}^2\|\bm{v}_{n,k}^t - \bm{v}_n^t\|^2
\end{align}\end{small}
The first inequality can be proved by employing the triangle inequality, $\|a+b\|^2 \leq 2\|a\|^2 + 2\|b\|^2$, and the last inequality can be proved based on the Assumption.

Finally, by substituting inequality (\ref{lemma3-2}) into inequality (\ref{lemma3-1}) and taking the average of the loss function expectation across all clients, the proof is concluded.
\end{proof}

\begin{lemma}
\label{e: lemma4}
$\sigma_u^2, \delta^2, \sigma_v^2, \rho^2$ are all defined in Assumptions, so we have  
\begin{small}\begin{align}
\nonumber
&\frac{L_{u}+L_{uv}}{2}\left\|\bm{u}^{t + 1}-\bm{u}^{t}\right\|^{2} \\
&\leq \frac{\left(L_{u}+L_{uv}\right)\eta_{u}^{2}\tau^{2}}{2}\sigma_{u}^{2}+\frac{3\left(L_{u}+L_{uv}\right)\eta_{u}^{2}\tau^{2}}{2}\xi^{2} \\
\nonumber
&\qquad+\frac{3\left(L_{u}+L_{uv}\right)\eta_{u}^{2}\tau^{2}}{2}(1+\rho^{2})\left\|\nabla_{u}F_{n}(\bm{u}^{t},V^{t})\right\|^{2} \\
\nonumber
&\qquad+\frac{3\left(L_{u}+L_{uv}\right)L_{u}^{2}\eta_{u}^{2}\tau}{2N}\sum_{k = 0}^{\tau}\sum_{n = 1}^{N}\left\|\bm{u}_{n,k}^t-\bm{u}^{t}\right\|^{2} \\
\nonumber
&\qquad+\frac{3\left(L_{u}+L_{uv}\right)L_{uv}^{2}\eta_{u}^{2}\tau}{2N}\sum_{k = 0}^{\tau}\sum_{n = 1}^{N}\left\|\bm{v}_{n,k}-\bm{v}_{n}^{t}\right\|^{2}
\end{align}\end{small}
\end{lemma}

\begin{proof}
For any client, we use the unbiased estimator gradient $G_n$ and the independent dataset $z_{n,k}^t$ from parameters $\bm{u}_{n,k}^t$ and $ \bm{v}_{n,k}^t$, so we have
\begin{subequations}
\begin{small}\begin{align}
\nonumber
&\frac{L_{u}+L_{uv}}{2}\left\|\bm{u}^{t + 1}-\bm{u}^{t}\right\|^{2} \\
\label{lemma4-1}
& =\frac{L_{u}+L_{uv}}{2}\left\|\frac{\eta_{u}}{N}\sum_{n = 1}^{N}\sum_{k = 0}^{\tau}G_{n,u}(\bm{u}_{n,k}^{t},\bm{v}_{n,k}^{t},\delta_{n,k}^{t})\right\|^{2} \\
\label{lemma4-2}
&\leq\frac{\left(L_{u}+L_{uv}\right)\eta_{u}^{2}\tau^{2}}{2}\sigma_{u}^{2} \\
\nonumber
&\qquad+\frac{\left(L_{u}+L_{uv}\right)\eta_{u}^{2}\tau}{2}\sum_{k = 0}^{\tau}\left\|\frac{1}{N}\sum_{n = 1}^{N}\nabla_{u}F_{n}(\bm{u}_{n,k}^{t},\bm{v}_{n,k}^{t})\right\|^{2} \\
\label{lemma4-3}
&\leq\frac{\left(L_{u}+L_{uv}\right)\eta_{u}^{2}\tau^{2}}{2}\sigma_{u}^{2}+\frac{\left(L_{u}+L_{uv}\right)\eta_{u}^{2}\tau^{2}}{2}\delta^{2} \\
\nonumber
&\qquad+\frac{3\left(L_{u}+L_{uv}\right)\eta_{u}^{2}\tau^{2}(1+\rho^{2})}{2}\left\|\nabla_{u}F_{n}(\bm{u}^{t},V^{t})\right\|^{2} \\
\nonumber
&\qquad+\frac{3\left(L_{u}+L_{uv}\right)\eta_{u}^{2}\tau}{2N}\sum_{k = 0}^{\tau}\sum_{n = 1}^{N}L_{u}^{2}\left\|\bm{u}_{n,k}^t-\bm{u}^{t}\right\|^{2} \\
\nonumber
&\qquad+\frac{3\left(L_{u}+L_{uv}\right)\eta_{u}^{2}\tau}{2N}\sum_{k = 0}^{\tau}\sum_{n = 1}^{N}L_{uv}^{2}\left\|\bm{v}_{n,k}-\bm{v}_{n}^{t}\right\|^{2}
\end{align}\end{small}
\end{subequations}
where the equivalence of equation (\ref{lemma4-1}) stems from the cumulative effect of model parameter updates through multiple gradient-based iterations performed locally on client devices. 
The derivation of inequality(\ref{lemma4-2}) is more intricate, building upon (\ref{lemma4-1}) through the addition and subtraction of $\nabla_{u}F_{n}(\bm{u}^{t},V^{t})$, combined with the application of Jensen's inequality ($f\left(\frac{\sum_{n}^{N} x_n}{N}\right) \le \frac{1}{N}\sum_{n}^{N} f(x_{n})$), and ultimately incorporating Assumption \ref{ass: Bounded Inner Error} to yield the final result. 
Inequality (\ref{lemma4-3}) is derived by sequentially adding and subtracting $\nabla_{u} F_n\left(\bm{u}^t, \bm{v}_n^{t}\right)$ and $\nabla_{u} F_n\left(\bm{u}^t, \bm{V}^t\right)$ within the squared norm, followed by the application of mean inequality chain properties to decompose the squared norm, ultimately yielding the final expression.
\end{proof}
\end{appendices}

\begin{lemma}
\label{e: lemma5}
$\sigma_u^2, \sigma_v^2$ are all defined in Assumptions, so we have  
\begin{small}\begin{align}
\nonumber
&\frac{L_v + L_{uv}}{2N}\sum_{n = 1}^{N}\left\lVert \bm{v}_{n}^{t+1}-\bm{v}_{n}^t\right\rVert^2 \\
&\leq \frac{\left(L_{v}+L_{uv}\right)\eta_{v}^{2}\tau^2}{2}\sigma_{v}^{2} \\
\nonumber
&\qquad+\frac{3\left(L_{v}+L_{uv}\right)\eta_{v}^{2}\tau^{2}}{2N}\left\|\nabla_{v}F_{n}(\bm{u}^{t},\bm{v}_{n}^{t})\right\|^{2} \\
\nonumber
&\qquad+\frac{3\left(L_{v}+L_{uv}\right)L_{vu}^{2}\eta_{v}^{2}\tau}{2N}\sum_{k = 0}^{\tau}\sum_{n = 1}^{N}\left\|\bm{u}_{n,k}^t-\bm{u}^{t}\right\|^{2} \\
\nonumber
&\qquad+\frac{3\left(L_{v}+L_{uv}\right)L_{v}^{2}\eta_{v}^{2}\tau}{2N}\sum_{k = 0}^{\tau}\sum_{n = 1}^{N}\left\|\bm{v}_{n,k}-\bm{v}_{n}^{t}\right\|^{2}
\end{align}\end{small}
\end{lemma}

\begin{proof}
For any client, we use the unbiased estimator gradient $G_i$ and the independent dataset $z_{n,k}^t$ from parameters $\bm{u}_{n,k}^t$ and $ \bm{v}_{n,k}^t$. 
Due to the aggregation mechanism, we can state that $\bm{v}^{t+1}_n=\bm{v}^{t}_{n,\tau}$, it suffices to demonstrate 
\begin{subequations}
\begin{small}\begin{align}
\nonumber
&\frac{L_v + L_{uv}}{2N}\sum_{n = 1}^{N}\left\lVert \bm{v}_{n}^{t+1}-\bm{v}_{n}^t\right\rVert^2\\
\nonumber
&=\frac{L_v + L_{uv}}{2N}\sum_{n = 1}^{N}\left\lVert \bm{v}_{n,\tau}^t-\bm{v}_{n}^t\right\rVert^2\\
\label{lemma5-1}
&=\frac{L_v + L_{uv}}{2N}\sum_{n = 1}^{N}\left\lVert \eta_v\sum_{k = 0}^{\tau}G_{n,v}\left(\bm{u}_{n,k}^{t},\bm{v}_{n,k}^{t},\delta_{n,k}^{t}\right)\right\rVert^2\\
\nonumber
&\leq\frac{\left(L_v + L_{uv}\right)\eta_v^2\tau}{2N}\sum_{n = 1}^{N}\sum_{k = 0}^{\tau}\left\lVert G_{n,v}\left(\bm{u}_{n,k}^{t},\bm{v}_{n,k}^{t},\xi_{n,k}^{t}\right)\right\rVert^2\\
\label{lemma5-2}
&\leq\frac{\left(L_v + L_{uv}\right)\eta_v^2\tau}{2}\sigma^2  \\
\nonumber
&\qquad+\frac{\left(L_v + L_{uv}\right)\eta_v^2\tau}{2N}\sum_{n = 1}^{N}\sum_{k = 0}^{\tau}\left\lVert \nabla_{v}F_n\left(\bm{u}_{n,k}^{t},\bm{v}_{n,k}^{t}\right)\right\rVert^2 \\
\label{lemma5-3}
&\leq\frac{\left(L_v + L_{uv}\right)\eta_v^2\tau}{2}\sigma_{v}^2 \\
\nonumber
&\qquad+\frac{3\left(L_v + L_{uv}\right)L_{vu}^2\eta_v^2\tau}{2N}\sum_{n = 1}^{N}\sum_{k = 0}^{\tau}\left\lVert \bm{u}_{n,k}^{t}-\bm{u}^{t}\right\rVert^2 \\
\nonumber
&\qquad+\frac{3\left(L_v + L_{uv}\right)L_{v}^2\eta_v^2\tau}{2N}\sum_{n = 1}^{N}\sum_{k = 0}^{\tau}\left\lVert \bm{v}_{n,k}^{t}-\bm{v}_{n}^{t}\right\rVert^2 \\
\nonumber
&\qquad+\frac{3\left(L_v + L_{uv}\right)\eta_v^2\tau^{2}}{2N}\sum_{n = 1}^{N}\left\lVert \nabla_{v}F_n\left(\bm{u}^{t},\bm{v}_{n}^{t}\right)\right\rVert^2
\end{align}\end{small}
\end{subequations}
Both of inequality (\ref{lemma5-1}) and inequality (\ref{lemma5-2}) follow the same derivation process as Lemma \ref{e: lemma4}.
Inequality (\ref {lemma5-3}) is derived by first successively adding and subtracting $\nabla_{v} F_n\left(\bm{u}^t, \bm{v}_{n,k}^{t}\right)$ and $\nabla_{v} F_n\left(\bm{u}^t, \bm{v}_{n}^t\right)$ within the squared norm. Subsequently, the properties of the mean - inequality chain are applied to decompose the squared norm, ultimately resulting in the final expression.
\end{proof}

\begin{lemma}
\label{e: lemma6}
On the premise of satisfying $\eta_{u} \le \frac{1}{\sqrt{6}\tau L_{u}\max\left\{1,\frac{L_{vu}L_{uv}}{L_{u}L_{v}}\right\}}$ and $\eta_{v} \le \frac{1}{\sqrt{6}\tau L_{v}\max\left\{1,\frac{L_{vu}L_{uv}}{L_{u}L_{v}}\right\}}$, we can obtain:
\begin{small}\begin{align}
    \nonumber
    &\bm{A}\sum_{n = 1}^{N}\sum_{k = 0}^{\tau}\left\lVert \bm{u}_{n,k}^{t}-\bm{u}_{n}^t\right\rVert^2+\bm{B}\sum_{n = 1}^{N}\sum_{k = 0}^{\tau}\left\lVert \bm{v}_{n,k}^t-\bm{v}_{n}^t\right\rVert^2 \\
    &\le\sum_{n = 1}^{N} (e^{2}-1)  (\bm{A} \tau\eta_{u}^{2}\sigma_{u}^{2}+\bm{B} \tau\eta_{v}^{2}\sigma_{v}^{2} \\
    \nonumber
    & \qquad 3\tau\eta_{u}^{2}\bm{A}\left\|\nabla_{u}F_{n}\left(\bm{u}_{n}^{t},\bm{v}_{n}^{t}\right)\right\|^{2}+3\tau\eta_{v}^{2}\bm{B}\left\|\nabla_{v}F_{n}\left(\bm{u}_{n}^{t},\bm{v}_{n}^{t}\right)\right\|^{2})
\end{align}\end{small}
\end{lemma}

\begin{proof}
First, we attempt to decompose $\bm{u}_{n, k - 1}^{t}$ and $\bm{v}_{n, k - 1}^{t}$ using formula $\bm{u}_{n,k}^{t}$ and $\bm{v}_{n,k}^t$, 
\begin{small}\begin{align}
    \nonumber
    &\left\lVert \bm{u}_{n,k}^{t} - \bm{u}_{n}^{t} \right\rVert^{2} \\
    \nonumber
    &= \left\lVert \bm{u}_{n,k - 1}^{t} - \eta_{u}G_{n,u}(\bm{u}_{n,k}^{t}, \bm{v}_{n,k}^{t}, \xi_{n,k}^{t}) - \bm{u}_{n}^{t} \right\rVert^{2} \\
    \nonumber
    &= \left\lVert \bm{u}_{n,k - 1}^{t} - \bm{u}_{n}^{t} \right\rVert^{2} + \left\lVert \eta_{u}G_{n,u}(\bm{u}_{n,k}^{t}, \bm{v}_{n,k}^{t}, \xi_{n,k}^{t}) \right\rVert^{2} \\ 
    \nonumber
    &\qquad\qquad- 2\langle \bm{u}_{n,k - 1}^{t} - \bm{u}_{n}^{t}, \eta_{u}G_{n,u}(\bm{u}_{n,k}^{t}, \bm{v}_{n,k}^{t}, \xi_{n,k}^{t}) \rangle\\
    \nonumber
    &= \left(1 + \frac{1}{\tau - 1}\right)\left\lVert \bm{u}_{n,k - 1}^{t} - \bm{u}_{n}^{t} \right\rVert^{2} \\
    \nonumber
    &\qquad\qquad+ \tau\eta_{u}^{2}\left\lVert G_{n,u}(\bm{u}_{n,k}^{t}, \bm{v}_{n,k}^{t}, \xi_{n,k}^{t}) \right\rVert^{2}\\
    &\leq \left(1 + \frac{1}{\tau - 1}\right)\left\lVert \bm{u}_{n,k - 1}^{t} - \bm{u}_{n}^{t} \right\rVert^{2} \\
    \nonumber
    &\qquad\qquad+ \tau\eta_{u}^{2}\left\lVert \nabla_{u}F_{n}(\bm{u}_{n,k}^{t}, \bm{v}_{n,k}^{t}) \right\rVert^{2} + \tau\eta_{u}^{2}\sigma_{u}^{2}\\
    \nonumber
    &\leq \left(1 + \frac{1}{\tau - 1}\right)\left\lVert \bm{u}_{n,k - 1}^{t} - \bm{u}_{n}^{t} \right\rVert^{2} + 3\tau\eta_{u}^{2}\left\lVert \nabla_{u}F_{n}(\bm{u}_{n}^{t}, \bm{v}_{n}^{t}) \right\rVert^{2} \\
    \nonumber
    &+ 3\tau\eta_{u}^{2}L_{u}^{2}\left\lVert \bm{u}_{n,k - 1}^{t} - \bm{u}_{n}^{t} \right\rVert^{2} + 3\tau\eta_{u}^{2}L_{uv}^{2}\left\lVert \bm{v}_{n,k}^{t} - \bm{v}_{n}^{t} \right\rVert^{2}+ \tau\eta_{u}^{2}\sigma_{u}^{2}
\end{align}\end{small}
Similarly, we can obtain the result of $\bm{v}_{n,k}^t$. By combining these results, we have:
\begin{small}\begin{align}
    \nonumber
    &\mathbf{E}\left\{\bm{A}\left\lVert \bm{u}_{n,k}^{t}-\bm{u}_{n}^t\right\rVert^2+\bm{B}\left\lVert \bm{v}_{n,k}^t-\bm{v}_{n}^t\right\rVert^2\right\} \\
    &\leq\left(1 + \frac{1}{\tau - 1}\right)\bm{A}\left\|\bm{u}_{n, k - 1}^{t}-\bm{u}_{n}^{t}\right\|^{2} \\
    \nonumber
    &+\left(1+\frac{1}{\tau - 1}\right)\bm{B}\left\|\bm{v}_{n, k-1}^{t}-\bm{v}_{n}^{t}\right\|^{2}\\
    \nonumber
    &+\bm{A}\left[\tau\eta_{u}^{2}\sigma_{u}^{2}+3\tau\eta_{u}^{2}\left\|\nabla_{u}F_{n}\left(\bm{u}_{n}^{t},\bm{v}_{n}^{t}\right)\right\|^{2}\right]\\
    \nonumber
    &+\bm{A}\left[3\tau\eta_{u}^{2}L_{u}^{2}\left\|\bm{u}_{n, k - 1}^{t}-\bm{u}_{n}^{t}\right\|^{2}+3\tau\eta_{u}^{2}L_{uv}^{2}\left\|\bm{v}_{n, k - 1}^{t}-\bm{v}_{n}^{t}\right\|\right] \\
    \nonumber
    &+\bm{B}\left[\tau\eta_{v}^2\sigma_{v}^{2}+3\tau\eta_{v}^{2}\left\|\nabla_{v}F_{n}\left(\bm{u}_{n}^{t},\bm{v}_{n}^{t}\right)\right\|^{2}\right] \\
    \nonumber
    &+\bm{B}\left[3\tau\eta_{v}^{2}L_{v}^{2}\left\|\bm{v}_{n, k - 1}^{t}-\bm{v}_{n}^{t}\right\|^{2}+3\tau\eta_{v}^{2}L_{vu}^{2}\left\|\bm{u}_{n, k - 1}^{t}-\bm{u}_{n}^{t}\right\|\right]\\
    \nonumber
    &=\left(1+\frac{1}{\tau - 1}\right)\left[\bm{A}\left\|\bm{u}_{n, k - 1}^{t}-\bm{u}_{n}^{t}\right\|^{2}+\bm{B}\left\|\bm{v}_{n, k}^{t}-\bm{v}_{n}^{t}\right\|^{2}\right] \\
    \nonumber
    &+\left[3\tau\eta_{u}^{2}L_{u}^{2}\bm{A}+3\tau\eta_{v}^{2}L_{vu}^{2}\bm{B}\right]\left\|\bm{u}_{n, k - 1}^{t}-\bm{u}_{n}^{t}\right\|^{2} \\
    \nonumber
    &+\left[3\tau\eta_{v}^{2}L_{v}^{2}\bm{B}+3\tau\eta_{u}^{2}L_{uv}^{2}\bm{A}\right]\left\|\bm{v}_{n, k - 1}^{t}-\bm{v}_{n}^{t}\right\|^{2} \\
    \nonumber
    &+3\tau\eta_{u}^{2}\bm{A}\left\|\nabla_{u}F_{n}\left(\bm{u}_{n}^{t},\bm{v}_{n}^{t}\right)\right\|^{2}+3\tau\eta_{v}^{2}\bm{B}\left\|\nabla_{v}F_{n}\left(\bm{u}_{n}^{t},\bm{v}_{n}^{t}\right)\right\|^{2} \\
    \nonumber
    &+\bm{A} \tau\eta_{u}^{2}\sigma_{u}^{2}+\bm{B} \tau\eta_{v}^{2}\sigma_{v}^{2}
\end{align}\end{small}
where we denote that $\bm{A}=L_{u}^2\eta_{vu}+L_{vu}^2\eta_v$ and $\bm{B}=L_{v}^2\eta_{v}+L_{uv}^2\eta_u$ for the sake of concise notation.

Because of Lemma \ref{e: lemma7}, when $\eta_{v} \le \frac{1}{\sqrt{6}\tau L_{v}\max\left\{1,\frac{L_{vu}L_{uv}}{L_{u}L_{v}}\right\}}$ and $\eta_{v} \le \frac{1}{\sqrt{6}\tau L_{v}\max\left\{1,\frac{L_{vu}L_{uv}}{L_{u}L_{v}}\right\}}$ can also be meet, we have
\begin{small}\begin{align}
\label{lemma6-1}
    \nonumber
    &\mathbf{E}\left\{\bm{A}\left\lVert \bm{u}_{n,k}^{t}-\bm{u}_{n}^t\right\rVert^2+\bm{B}\left\lVert \bm{v}_{n,k}^t-\bm{v}_{n}^t\right\rVert^2\right\} \\
    \nonumber
    &\le\left(1+\frac{2}{\tau - 1}\right)\left[\bm{A}\left\|\bm{u}_{n, k - 1}^{t}-\bm{u}_{n}^{t}\right\|^{2}+\bm{B}\left\|\bm{v}_{n, k}^{t}-\bm{v}_{n}^{t}\right\|^{2}\right] \\
    \nonumber
    &+3\tau\eta_{u}^{2}\bm{A}\left\|\nabla_{u}F_{n}\left(\bm{u}_{n}^{t},\bm{v}_{n}^{t}\right)\right\|^{2}+3\tau\eta_{v}^{2}\bm{B}\left\|\nabla_{v}F_{n}\left(\bm{u}_{n}^{t},\bm{v}_{n}^{t}\right)\right\|^{2} \\
    &+\bm{A} \tau\eta_{u}^{2}\sigma_{u}^{2}+\bm{B} \tau\eta_{v}^{2}\sigma_{v}^{2}
\end{align}\end{small}

After aggregating the parameters following $k=0, ..., \tau$ rounds of iteration, we can obtain
\begin{small}\begin{align}
    \mathbf{E}\left\{\bm{A}\left\lVert \bm{u}_{n,k}^{t}-\bm{u}_{n}^t\right\rVert^2+\bm{B}\left\lVert \bm{v}_{n,k}^t-\bm{v}_{n}^t\right\rVert^2\right\} \le \sum_{j=0}^{k-1} \left(1+\frac{2}{\tau-1}\right)^{j}C
\end{align}\end{small}
where $C$ represents the abbreviated form of the last four terms of inequality (\ref{lemma6-1}).

So we have
\begin{small}\begin{align}
    &\bm{A}\sum_{n = 1}^{N}\sum_{k = 0}^{\tau}\left\lVert \bm{u}_{n,k}^{t}-\bm{u}_{n}^t\right\rVert^2+\bm{B}\sum_{n = 1}^{N}\sum_{k = 0}^{\tau}\left\lVert \bm{v}_{n,k}^t-\bm{v}_{n}^t\right\rVert^2 \\
    \nonumber
    &\le \sum_{n = 1}^{N}\sum_{k = 0}^{\tau}\frac{\tau-1}{2}\left(1+\frac{2}{\tau-1}\right)^k C \\
    \label{lemma6-2}
    &=\sum_{n = 1}^{N}\left[\left(1+\frac{2}{\tau-1}\right)^{\tau}-1\right] C \\
    \label{lemma6-3}
    &\le\sum_{n = 1}^{N} (e^{2}-1) C 
\end{align}\end{small}
where inequality (\ref{lemma6-2}) and inequality (\ref{lemma6-3}) are derived by using the geometric summation formula and the series summation formula ($(1+\frac{1}{x})^x \le e$) respectively.
\end{proof}

\begin{lemma}
\label{e: lemma7}
We denote that $\bm{A}=L_{u}^2\eta_{u}+L_{vu}^2\eta_{v}$ and $\bm{B}=L_{v}^2\eta_{v}+L_{uv}^2\eta_{u}$. 
We suppose that $\eta_{u} \le \frac{1}{\sqrt{6}\tau L_{u}\max\left\{1,\frac{L_{vu}L_{uv}}{L_{u}L_{v}}\right\}}$ and $\eta_{v} \le \frac{1}{\sqrt{6}\tau L_{v}\max\left\{1,\frac{L_{vu}L_{uv}}{L_{u}L_{v}}\right\}}$, the following inequalities can be obtained:
\begin{small}\begin{align}
\label{lemma7-1}
3\tau^2\eta_{u}^{2}L_{u}^{2}\bm{A}+3\tau^2\eta_{v}^{2}L_{vu}^{2}\bm{B} \le \bm{A} \\
\label{lemma7-2}
3\tau^2\eta_{v}^{2}L_{v}^{2}\bm{B}+3\tau^2\eta_{u}^{2}L_{uv}^{2}\bm{A} \le \bm{B}
\end{align}\end{small}
\end{lemma}

\begin{proof}
To prove the validity of inequality (\ref{lemma7-1}), under the given assumptions about $\eta_u$, if both $3\tau^2\eta_{u}^{2}L_{u}^{2}\bm{A} \le \frac{\bm{A}}{2} $ and $ 3\tau^2\eta_{v}^{2}L_{vu}^{2}\bm{B} \le \frac{\bm{A}}{2}$ strictly satisfy the specified requirements, then inequality (\ref{lemma7-1}) is guaranteed to hold. 
Regarding $ 3\tau^2\eta_{v}^{2}L_{vu}^{2}\bm{B} \le \frac{\bm{A}}{2}$, we have the following derivations by plugging into $\bm{A}$ and $\bm{B}$:
\begin{small}
\begin{align}
    6\tau^{2}\eta_{v}^{2}L_{vu}^{2}\left(L_{v}^2\eta_{v}+L_{uv}^2\eta_{u}\right) \le L_{u}^2\eta_{u}+L_{vu}^2\eta_{v}
\end{align}
\end{small}
Then, the range of $\eta_{v}$ satisfying the requisite conditions can be derived by applying the following inequality:
\begin{small}
\begin{align}
    6\tau^{2}L_{vu}^{2}L_{v}^2\eta_{v}^{3} &\le L_{vu}^2\eta_{v} \\
    6\tau^{2}L_{vu}^{2}L_{uv}^2\eta_{u}\eta_{v}^{2} &\le L_{u}^2\eta_{u}
\end{align}
\end{small}
Combining the results of above two inequalities, we can conclude that $\eta_{v} \le \frac{1}{\sqrt{6}\tau L_{v}\max\left\{1,\frac{L_{vu}L_{uv}}{L_{u}L_{v}}\right\}}$ can satisfy inequality (\ref{lemma7-1}).
And using the same method, we can derive that $\eta_{u} \le \frac{1}{\sqrt{6}\tau L_{u}\max\left\{1,\frac{L_{vu}L_{uv}}{L_{u}L_{v}}\right\}}$ also meet requirement of inequality (\ref{lemma7-2}).
\end{proof}

\end{document}